\documentclass[final,11pt]{elsarticle}
\usepackage[left=1in,right=1in,top=1in,bottom=1in]{geometry}

\journal{}

\usepackage{amsthm,amsmath,amssymb,amsfonts}
\usepackage{mathrsfs}
\usepackage{mathbbol}
\usepackage{enumitem}
\usepackage{algorithm}
\usepackage{algorithmic}
\usepackage{float}
\usepackage{mathtools}
\usepackage{subcaption}
\usepackage{relsize}
\usepackage{graphicx}
\usepackage{booktabs}
\usepackage{cite}

\usepackage{setspace}

\DeclareMathAlphabet{\mathup}{OT1}{\familydefault}{m}{n}
\newcommand{\ip}[2]{\langle{#1}, {#2}\rangle}
\newcommand{\NqNisp}{{N_\text{quad}^\text{nisp}}}
\newcommand{\Nord}{{N_\text{ord}}}
\newcommand{\Np}{{N_\text{p}}}
\newcommand{\Npc}{{N_\text{PC}}}
\newcommand{\Nq}{{N_\text{quad}}}

\newcommand{\Nkl}{{N_\text{kl}}}
\newcommand{\norm}[1]{\| {#1} \|}
\newcommand{\ffn}[1]{f^{({#1})}}
\newcommand{\fn}{f^{(n)}}

\renewcommand{\vec}[1]{{\mathchoice
                     {\mbox{\boldmath$\displaystyle{#1}$}}
                     {\mbox{\boldmath$\textstyle{#1}$}}
                     {\mbox{\boldmath$\scriptstyle{#1}$}}
                     {\mbox{\boldmath$\scriptscriptstyle{#1}$}}}}

\newcommand{\trace}{\mathrm{Tr}}

\newcommand{\eps}{\mbox{\large\(\varepsilon\)}}

\newcommand{\E}[1]{\mathbb{E}\left\{ {#1} \right\}}
\newcommand{\EE}[2]{\mathbb{E}_{{#1}}\left\{ {#2} \right\}}
\newcommand{\var}[1]{\mathbb{V}\left\{ {#1} \right\}}
\newcommand{\C}{\mathcal{C}}
\newtheorem{theorem}{Theorem}[section]
\newtheorem{proposition}{Proposition}[section]
\newtheorem{remark}{Remark}[section]

\newtheorem{thm}{Theorem}[section]
\newtheorem{cor}[thm]{Corollary}
\newtheorem{lem}{Lemma}
\newtheorem{prop}[thm]{Proposition}
\newcommand{\R}{\mathbb{R}}
\newcommand{\Dtot}{D_\text{tot}}
\newcommand{\Stot}{S_\text{tot}}
\newcommand{\gen}[1]{\mathfrak{{#1}}}
\newcommand{\gStot}{\gen{S}_\text{tot}}

\newcommand{\xx}{\vec{\xi}}
\newcommand{\xu}{{\xx_U}}
\newcommand{\zz}{{\xx_{U^\complement}}}
\newcommand{\zznom}{{\bar\xx_{U^\complement}}}
\newcommand{\prob}{\mathbb{P}}
\newcommand{\Oz}{{\Omega_2}}
\newcommand{\tran}{{\mkern-1.5mu\mathsf{T}}}                

\DeclareMathSymbol{\ast}{\mathbin}{symbols}{"03}

\allowdisplaybreaks
\begin{document}

\begin{frontmatter}

\author[label1]{Alen Alexanderian}
\address[label1]{Department of Mathematics, North Carolina State University, Raleigh, NC}

\ead{alexanderian@ncsu.edu}
\author[label1]{Pierre A.~Gremaud}
\ead{gremaud@ncsu.edu}

\author[label1]{Ralph C.~Smith}
\ead{rsmith@ncsu.edu}

\title{Variance-based sensitivity analysis for time-dependent processes}

\begin{abstract}
The global sensitivity analysis of time-dependent processes requires
history-aware approaches. We develop for that purpose a variance-based method
that leverages the correlation structure of the problems under study and
employs surrogate models to accelerate the computations.   The errors resulting
from fixing unimportant uncertain parameters to their nominal values are
analyzed through a priori estimates.   We illustrate our approach  on a
harmonic oscillator example and on a nonlinear dynamic cholera model.  
\end{abstract}


\end{frontmatter}

\section{Introduction}

%
%
The ability to make reliable predictions from time-dependent mathematical models of the form 
\begin{equation}\label{equ:basic_model}
Y = f(t, \xx), \quad t \in [0, T],
\end{equation}
where $\xx \in \R^\Np$ is a vector of uncertain model parameters, relies
crucially on understanding and  quantifying the impact  of $\xx$ on $f$. One
approach, due to Sobol', consists in apportioning to each element (or
group of elements) of $\xx$  its contribution to the variance of $f$
\citep{Sobol93,Sobol01,Owen13,SaltelliRattoAndresEtAl08}.  Such a global
sensitivity analysis (GSA) enables focusing  computational resources on
quantifying the uncertainties in the elements of $\xx$ that are most
influential on the variability of $f$. 
The present work is about extending Sobol's approach to problems with 
time-dependent outputs.

Most of the literature on global sensitivity analysis considers scalar outputs as opposed to the functional framework corresponding to (\ref{equ:basic_model}).
This amounts, for instance, to analyzing the  sensitivity  
of $f(t_0, \xx)$  for a fixed $t_0$ or to the study of integrated 
quantities such as $y(\xx) = \int_0^T f(t, \xx) \,
dt$. 
While it is  possible to apply Sobol's approach pointwise in time
\citep{AlexanderianWinokurSrajEtAl12,Namhata2016OladyshkinDilmoreEtAl16}, 
for instance at the nodes of a grid, 
\begin{equation}\label{equ:time_grid}
   0 = t_0 < t_1 < \cdots < t_{n-1} < t_n = T,
\end{equation} 
this approach presents two shortcomings. First, treating the $f(t_k,
\xx)$'s, $k=1, \dots, n$,  independently of one another ignores the
temporal correlation structure of the process.  Second, the variance of the process
itself varies in time therefore skewing relative importance measurements across
time.   More precisely, a ``yardstick" is needed at each time to determine
the influential parameters at that time; for the standard Sobol's indices,
this yardstick is the variance of $f$ at the corresponding time. When the
yardstick changes with time, confusion ensues: how to compare carrying a small
portion of a large variance with a large portion of a small one? These delicate scaling issues are also present in derivative-based sensitivity
analysis.  

As an illustrative example, consider an underdamped mechanical oscillator whose motion is 
governed by the initial value problem
\begin{equation}\label{equ:IVP}
\begin{aligned}
   &y'' + 2 \alpha y' + (\alpha^2 + \beta^2) y = 0,\\ 
   &y(0)  = \ell, \quad 
   y'(0) = 0.
\end{aligned}
\end{equation}
The solution is 
\begin{equation}\label{equ:pendulum_sol}
y(t; \alpha, \beta, \ell) = \ell e^{-\alpha t} (\cos \beta t + \frac{\alpha}{\beta} \sin \beta t),
\end{equation}
and the corresponding process  is given by $f(t, \xx) = y(t; \xx)$,
where $\xx$ is a random vector that parameterizes the uncertainty in the parameters
$(\alpha, \beta, \ell)$.
Figure~\ref{fig:pendulum_init} (left)  shows the time evolution of the 
mean trajectory (solid line) and the two standard deviation bounds (dashed lines). 
The values of the traditional pointwise total Sobol' indices, which we recall in Section~\ref{sec:var}, are reported in Figure~\ref{fig:pendulum_init} (right).
 \begin{figure}[h]\centering
\includegraphics[width=.45\textwidth]{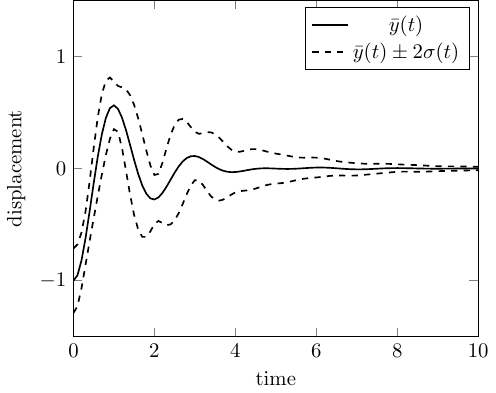}
\includegraphics[width=.45\textwidth]{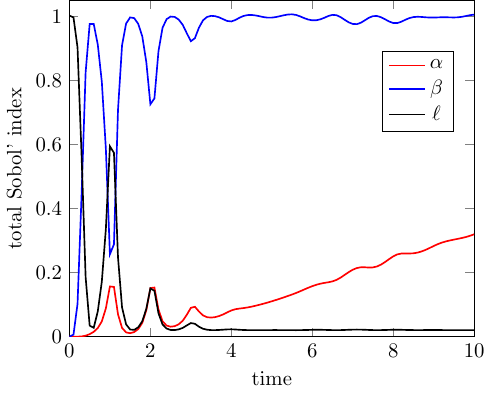}
\caption{Behavior of the mechanical oscillator problem (\ref{equ:IVP}) with uncertain parameters  $\alpha \sim \mathcal U(3/8, 5/8)$, 
$\beta \sim \mathcal U(10/4, 15/4)$  and $\ell \sim \mathcal U(-5/4,-3/4)$. Left: mean trajectory $\bar{y}(t)$ and  
the two standard deviation bounds, obtained via Monte Carlo sampling in the
uncertain parameter space. Right: standard Sobol' indices over time.
}
\label{fig:pendulum_init}
\end{figure}
These results are difficult to interpret as the balance of sensitivities changes
multiple times. Moreover,  this standard approach is entirely unaware of the
history of the process and, specifically here, of the asymptotically
diminishing variance.  For instance,  the  reported increasing influence of
$\alpha$ is largely an artifact of the method.   
We revisit this example throughout the article and, 
in Section~\ref{sec:cholera}, investigate similar issues on a more involved
dynamical system modeling  the spread of cholera.

Porting variance based GSA methods from the scalar case to the vectorial case or, more generally, to the functional case corresponding to   
(\ref{equ:basic_model}) presents three challenges:

\begin{enumerate}
\item[I.] the yardstick issue,
\item[II.] the need for a functional framework allowing analysis and method development,
\item[III.] the need for a computational framework allowing the implementation of efficient algorithms on 
realistic problems and applications. 
\end{enumerate} 

Our work is motivated by  \citep{GamboaJanonKleinEtAl14} which resolves challenge I in a general setting. 
We offer here key contributions to the resolutions of challenges II and III.  

\begin{itemize}
\item We establish in Section~\ref{sec:method} the equivalence of  the indices
from \citep{GamboaJanonKleinEtAl14}  to those previously proposed
in~\citep{LamboniMonodMakowski11}; see Theorem~\ref{thm:main}.  

\item We analyze, in the functional case,  the effect
of fixing inessential variables as determined by computing generalized Sobol'
indices.  This is done theoretically in Section~\ref{sec:unimp} and is
thoroughly investigated in the numerical results in Section~\ref{sec:cholera}.  
\item We leverage our representation of the generalized indices to enable the analysis and efficient implementation of two distinct approaches for the computation of these indices, namely surrogate models (see Section~\ref{sec:pointwise})  and  spectral representations (see Section~\ref{sec:spectral}).   
\item We present
comprehensive numerical results that provide insight into sensitivity analysis
of time-dependent processes; additionally, our numerical results examine
various aspects of our methods and show their effectiveness. 
\end{itemize}

We conclude this Introduction with a brief overview of surrogate models and spectral representations in the context of functional GSA. 

Surrogate models such as polynomial chaos (PC)
expansions~\citep{Wiener:1938,LeMaitreKnio10,Xiu10}, multivariate adaptive
regression splines (MARS)~\citep{friedman93}, and Gaussian processes have become
increasingly popular tools in uncertainty quantification literature; the
references~\citep{CrestauxLeMaitreMartinez09,Sudret08,BlatmanSudret10,
Alexanderian13,HartAlexanderianGremaud17,kleijnen,legratiet} provide a
non-exhaustive sample of the literature on their use for variance-based
sensitivity analysis.  These approaches replace repeated solutions of
computationally expensive models by inexpensive evaluations of a surrogate
model.  They can provide orders of magnitude speedups. While our approach is agnostic the choice of surrogates, we use PC surrogates in our implementation. 
In practice,  the choice of surrogate model  should be based on the demands of the problem at hand.
In the present context, a surrogate model $\tilde{f}(t_k, \xx)
\approx f(t_k, \xx)$ can be  constructed for every $t_k$ in the grid~\eqref{equ:time_grid} and generalized indices can subsequently be approximated at
negligible computational cost (see again Section~\ref{sec:pointwise}).  However, the full approximating power of many state-of-the-arts surrogates can only be harvested at the price of optimizing them at each specific  time $t_k$. This may be prohibitively expensive and  intractable especially when the number of time steps $n$ is large.  

This observation motivates the consideration of  spectral representations and, specifically here, the Karhunen--Lo\`{e}ve (KL) expansion 
\[
    f(t, \xx) \approx f_0(t) + \sum_{j=1}^\Nkl f_i(\xx) e_i(t),
\]
of the process $f$, where $f_0$ is the mean of the process, the 
$f_i(\xx)$'s are expansion coefficients (see Section~\ref{sec:spectral}) with 
variance $\var{f_i} = \lambda_i$ 
and where $\lambda_i$ is an eigenvalue of the covariance operator of $f$ with corresponding 
eigenvector  $e_i$.
The \emph{modes}  $\{f_i\}_{i=1}^\Nkl$ encode the uncertainty in $f$ and the dynamics of the
process is quantified by the superposition of the dominant
eigenvectors $\{ e_i \}_{i = 1}^\Nkl$. 
For processes with  fast decaying eigenvalues $\lambda_i$---which is often observed
in applications---a small truncation level $\Nkl$ can be used, i.e., the process can be represented by a small number of modes. 

The principle and feasibility  of functional GSA based on spectral representations are explored in~\citep{CampbellMcKayWilliams06}. 
These ideas are picked up in~\citep{LamboniMonodMakowski11} (see also~\citep{HongLuyi16})  where 
aggregate Sobol' indices are proposed for vectorial and functional outputs, 
based on the KL expansion of $f$. The unifying theoretical framework of Section~\ref{sec:spectral} is here completed by a thorough discussion of the computational issues linked to the use of spectral representations for functional GSA including the approximation of the covariance function via quadratures,  the computation of the spectral decomposition of the discretized  covariance operator and the use of polynomial surrogates for the KL modes, i.e.,  
 $\tilde f_i(\xx) \approx f_i(\xx)$, 
$i \in \{1, \ldots, \Nkl\}$; see Section~\ref{sec:computational}. 

The approach based on the KL expansion  not only provides an efficient
method for computing the generalized Sobol' indices, it is also structure revealing:  the uncertainty in the output can be captured
efficiently by the dominant KL modes of $f$.  Thus, this approach can also
guide the computation of efficient \emph{global in time} surrogate models to be used in the statistical study of $f$, beyond sensitivity analysis.

The benefits of
combining surrogate models such as PC expansions and random field representations using
KL expansions have been realized in other related works on GSA.
For example, the work~\citep{Pronzato19} presents a novel approach
for efficient computation of Sobol' indices for scalar outputs that combines
PC expansions and random field modeling using
KL expansions.

%
%
\section{Variance-based sensitivity indices for time-dependent processes}\label{sec:var}
For simplicity, we assume the uncertain parameters $\xi_1, \dots, \xi_{N_p}$
to be independent $\mathcal U(-1, 1)$ random variables.  Hence, we work in a
measure space $(\Omega, \mathcal{B}(\Omega), \mu)$, where $\Omega = [-1,
1]^\Np$, $\mathcal{B}(\Omega)$ is the Borel sigma-algebra on $\Omega$, and the
probability measure $\mu$ is the normalized $\Np$-dimensional Lebesgue measure
on $\Omega$: $\mu(d\xx) = 2^{-\Np} d\xx$.  It is straightforward to extend our
definitions and results to the case of any random vector $\xx$ with independent
elements.

We consider a random process $f:[0,T]\times \Omega \to \R$, and assume 
$f \in L^2([0,T]\times \Omega)$. Moreover, we assume $f$ to 
be mean-square continuous:
\begin{eqnarray}
    \lim_{h \to 0} \int_\Omega \big(f(t+h, \xx) - f(t, \xx)\big)^2 \, \mu(d\xx) = 0, \quad \text{for all } t \in [0, T]. \label{msc}
\end{eqnarray}
It follows that the mean $f_0(t) = \int_\Omega f(t, \xx) \, \mu(d\xx)$ and the covariance function 
\begin{equation}\label{equ:cov}
c(s, t) = \int_\Omega \big(f(s, \xx) - f_0(s)\big)\big(f(t, \xx) - f_0(t)\big)\, 
\mu(d\xx), \quad s, t \in [0, T], 
\end{equation}
are continuous on $[0,T]$ and $[0,T]\times[0,T]$ respectively~\citep[Theorem 7.3.2]{HsingEubank15},~\citep[Theorem 2.2.1]{Adler10}.  
In practice,
 the covariance function can be approximated through sampling
\begin{equation}\label{equ:cov_samp}
c(s, t) \!\approx\! c^N(s, t) \!=\! \frac{1}{N\!-\!1} \!\sum_{k=1}^{N} 
f_c(t, \xx^k) f_c(s, \xx^k), \quad f_c(t, \xx^k) \!=\! f(t, \xx^k) -  
\frac1N\!\sum_{j = 1}^N f(t, \xx^j). 
\end{equation}
Without loss of generality, it is possible to consider only centered processes, i.e., $f_0 \equiv 0$; we do so below.

\begin{remark}
We point out an important implication of the mean-square continuity assumption.
Assuming $f$ is mean-square continuous, we can conclude the existence of 
a modification\footnote{We say $f$ and $g$ are modifications of one another if 
for all $t \in [0, T]$,
$g(t, \cdot) = f(t, \cdot)$ almost surely.}
$g$ of $f$ such that $g$ is jointly measurable on the product space $(\Omega, \mathcal{B}(\Omega)) \otimes ([0, T], \mathcal{B}([0, T]))$;
see Proposition~3.2 in~\citep{DaPratoZabczyk14}. Note also,
\[
   \int_0^T \int_\Omega |g(t, \xx)|^2 \, \mu(d\xx) \, dt =  
   \int_0^T \int_\Omega |f(t, \xx)|^2 \, \mu(d\xx) \, dt 
   = \int_0^T c(t, t) \, dt < \infty.  
\] 
Therefore, as a consequence of Fubini's Theorem~\citep{big_Rudin}, $g \in L^2([0, T]
\times \Omega)$. Thus, 
due to the mean-square continuity assumption, 
replacing $f$ with a suitable modification, 
the requirement that $f \in L^2([0, T] \times \Omega)$ is satisfied. 
\end{remark}

\subsection{Sobol' indices}
Consider  the index set $X = \{1, \ldots, \Np\}$ and a subset $U = \{i_1, i_2, \ldots, i_s\} \subset X$.
We define  $\xu = (\xi_{i_1}, \xi_{i_2}, \ldots, \xi_{i_s})$ 
and  $\zz =  (\xi_{j_1}, \xi_{j_2}, \ldots, \xi_{j_{s'}})$ with $\{ j_1, j_2, \ldots, j_{s'}\} = X \setminus U = U^\complement$.
At each time $t$, we write $f$ according to its second-order ANOVA-like decomposition
\begin{eqnarray}
    f(t, \xx) =  f_U(t, \xu) + f_{U^\complement}(t, \zz) + f_{U,U^\complement}(t, \xx), \label{anova}
\end{eqnarray}
where
\[
\begin{aligned}
f_U(t, \xu) &:= \mathbb E\{f|\xu\},\\ 
f_{U^\complement}(t, \zz) &:= \mathbb E\{f|\zz\},\\
f_{U,U^\complement}(t, \xx) &:=  f(t, \xx) -  f_U(t, \xu) - f_{U^\complement}(t, \zz).
\end{aligned}
\]
The total variance $D(f;t)$ of $f$ can correspondingly be decomposed into 
\[
   D(f; t) = D^U(f; t) + D^{U^\complement}(f; t) + D^{U,{U^\complement}}(f; t), 
\]
where
\[
   D^U(f; t) := \EE{\xu}{f_U(t, \xu)^2}, \quad 
   D^{U^\complement}(f; t) := \EE{\zz}{f_{U^\complement}(t, \zz)^2}.
\]
The standard pointwise first and total Sobol' indices for $\xu$ are then defined by apportioning to the 
$\xu$ parameters their relative contribution to the variance of $f$   \citep{Sobol93,Sobol01} 
\begin{eqnarray}
    S^U(f; t) :=  \frac{D^{U}(f; t)}{D(f; t)}, \quad   
    \Stot^U(f; t) := \frac{\Dtot^U(f; t)}{D(f; t)}, \label{tradsobol}
\end{eqnarray}
where $\Dtot^{U}(f; t) := D^U(f; t) + D^{U,{U^\complement}}(f; t)$.
\subsection{Generalized Sobol' indices for time-dependent problems}
Pointwise in time indices such as (\ref{tradsobol}) ignore all time correlations. To characterize these correlations, 
we consider the covariance operator $\C:L^2([0, T]) \to L^2([0, T])$ of $f$,
\begin{equation}\label{equ:covariance_operator}
    \C[u](s) = \int_0^T c(s, t) u(t)\, dt,
\end{equation} 
where the covariance function $c$ is defined in~\eqref{equ:cov}; $\C$ is a trace-class positive selfadjoint 
operator with eigenvalues 
$\{ \lambda_i \}_{i = 1}^\infty$ and a complete  set of orthonormal eigenvectors $\{ e_i \}_{i = 1}^\infty$.
By Mercer's theorem~\citep{Mercer1909,Lax02}, we have
\begin{equation}\label{equ:mercer}
   c(s, t) = \sum_{j = 1}^\infty \lambda_j e_j(s) e_j(t),
\end{equation}
where the convergence of the infinite sum is uniform and absolute in $[0, T] \times [0, T]$.

Let $c_U$ and $\C_U$ be respectively the  covariance function and
covariance operator corresponding to $f_U$ from (\ref{anova}).
Following~\citep{GamboaJanonKleinEtAl14}, the generalized first order sensitivity
index for $\xu$ can be defined as
\begin{equation}\label{equ:generalized_tr}
\gen{S}^U(f; T) := \frac{\trace(\C_U)}{\trace(\C)}.
\end{equation}
The next result shows the generalized indices to be nothing but the ratio of the time integrals of the numerator and denominator of the standard indices. 
\begin{proposition}\label{prp:integrated_indices}
Let the random process $f$ be as above. Then,
\begin{equation}\label{equ:generalized}
\gen{S}^U(f; T) = \frac{\int_0^T D^{U}(f; t)\,dt}{\int_0^T D(f; t)\,dt}. 
\end{equation}
\end{proposition}
\begin{proof}
Considering the denominator in (\ref{equ:generalized}), we obtain
\begin{multline*}
\int_0^T D(f; t) \, dt 
= \int_0^T c(t, t) \, dt = \int_0^T \sum_{j = 1}^\infty \lambda_j e_j(t)^2 \, dt 
= \sum_{j = 1}^\infty \lambda_j \int_0^T e_j(t)^2 \, dt 
= \sum_{j = 1}^\infty \lambda_j = \trace(\C),
\end{multline*}
where the second equality follows from~\eqref{equ:mercer},
the interchange of integral and summation
is justified by the Monotone Convergence Theorem, and the third equality uses
the fact that eigenvectors are orthonormal. 
The numerator can be treated similarly. 
\end{proof}

The integrals in  (\ref{equ:generalized}) can be computed  via a quadrature formula
on $[0, T]$, with nodes $\{ t_m \}_{m = 1}^\Nq$ and weights $\{w_m\}_{m = 1}^\Nq$, yielding 
the approximation 
\begin{equation}\label{equ:gStot_discrete}
\gen{S}^U(f; T) \approx \frac{\sum_{m = 1}^\Nq w_m D^{U}(f; t_m) }{ \sum_{m = 1}^\Nq w_m D(f; t_m)}.
\end{equation}
The special case of equal weights and uniform time steps in (\ref{equ:gStot_discrete}) corresponds to the 
approach suggested in 
\citep{GamboaJanonKleinEtAl14} for sensitivity analysis for time-dependent processes.

Similarly to (\ref{equ:generalized}), we define
generalized total Sobol' indices as 
\begin{equation}\label{equ:newtotsobol}
\gStot^U(f; T) := \frac{\int_0^T \Dtot^{U}(f; t)\,dt}{\int_0^T D(f; t)\,dt}.
\end{equation}
We note that
$
\gStot^U(f; T) 
=1-\frac{\int_0^T D^{U^\complement}(f; t) \,dt}{\int_0^T D(f; t)\,dt}
=1-\gen{S}^{U^\complement}(f; T) = \frac{\trace(\C) - \trace(\C_{U^\complement})}{\trace(\C)}.
$
Further, as their pointwise counterparts, the generalized total Sobol' indices
$\gStot^U(f; T)$ admit an approximation theoretic interpretation; namely, for a
centered $f$
\[
\gStot^U(f; T)  = \frac{\|f - \mathcal P_{U^\complement} f \|_{L^2([0,T] \times \Omega)}^2}{\| f\|_{L^2([0,T] \times \Omega)}^2}
\]
where $\mathcal P_{U^\complement}$ is the orthogonal projector $\mathcal P_{U^\complement}:
L^2([0,T]; L^2(\Omega)) \to L^2([0,T]; V_{U^\complement})$ with $V_{U^\complement} \subset
L^2(\Omega)$ being, roughly speaking, the subspace of functions that do not
depend on $\xx_U$ in $L^2(\Omega)$, see \citep{hg} for full justification. 

If a fine Monte Carlo (MC) sampling of $f$ is feasible, the partial variances
appearing in the expressions for $\gen{S}^U(f; T)$ and $\gStot^U(f; T)$ can be
computed using traditional MC-based algorithms for estimating the pointwise
Sobol' indices; see e.g.,\citep{Sobol01}.  However, computing the generalized
indices via direct Monte Carlo sampling is in general expensive. This is due to
the need for a large number of function evaluations. In
section~\ref{sec:method}, we present efficient methods for computing these
indices using suitable approximations of $f$.

To illustrate the concepts introduced so far, we
return to the  mechanical oscillator example~\eqref{equ:IVP} and compute its
generalized total Sobol' indices; see Figure~\ref{fig:pendulum_sobol_timedep}
(left). Figure~\ref{fig:pendulum_sobol_timedep} (right) illustrates the
evolution of the generalized Sobol' indices over successively larger intervals.
These results provide a clear analysis of the relative importance of  the input
parameters along with a ``history aware" description of the evolution of these
relative importance measurements. While the pointwise in time Sobol' indices
show a significant growing influence of $\alpha$ over time, see again
Figure~\ref{fig:pendulum_init}(right), the generalized indices stabilize
quickly and provide an importance assessment of the variables that is
consistent over time.

\begin{figure}[ht]\centering
\includegraphics[width=.45\textwidth]{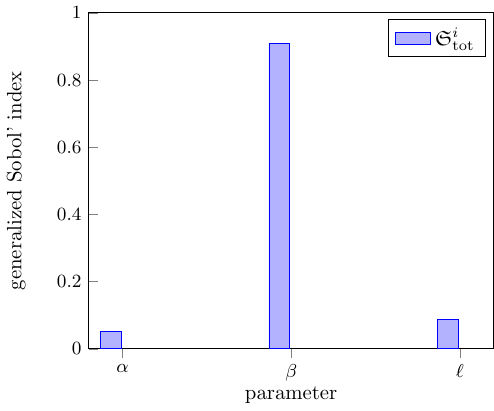}
\includegraphics[width=.45\textwidth]{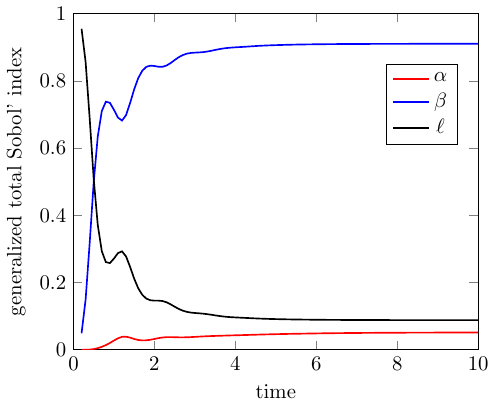}
\caption{Behavior of the mechanical oscillator problem (\ref{equ:IVP}) with uncertain parameters  $\alpha \sim \mathcal U(3/8, 5/8)$, 
$\beta \sim \mathcal U(10/4, 15/4)$  and $\ell \sim \mathcal U(-5/4,-3/4)$. Left: generalized Sobol' indices $\gStot^U(y; T)$, $U=\alpha, \beta, \ell$, with $T = 10$ and $y$ is the solution to (\ref{equ:IVP}).
Right: $\gStot^U(y; \tau)$ with $\tau \in (0, 10)$.} 
\label{fig:pendulum_sobol_timedep}
\end{figure}

%
%
\section{A priori estimates with fixed unimportant variables}\label{sec:unimp}
Suppose we have identified a subset $\xu$, $U \subset X = \{1, \ldots, \Np\}$, of parameters such that 
$\gStot^{U^\complement}(f; T)$ is small compared to $\gStot^U$. 
In other words, the parameters $\zz$ are unimportant and it should thus be possible to fix them at some nominal value $\zznom$ and consider the 
``reduced" function 
\[
 \bar f(t, \xu) = f(t, \xu, \zznom),
\]
as a reasonable approximation of $f$. The next result formalizes this line of thought by establishing a direct link between $\gStot^\zz(f; T)$ and a 
measure of the relative error attached to the approximation $f \approx \bar f$; this result generalizes to time dependent problems the work of 
\citep{SobolTarantolaGatelliEtAl06} on stationary problems.

\begin{prop}\label{prp:basic}
Let $\eps(f; t, \zznom) =  \frac12 \int  \left(f(t, \xx) - \bar f(t, \xu)\right)^2\, \mu(d\xx)$. Then 
\[
\mathscr E := \frac{\int_0^T \eps(f; t, \zznom)\,dt}{\int_0^T D(f; t) \, dt}
\]
provides a  measurement of the relative error linked to the approximation $f \approx \bar f$ and furthermore
\[
\mathbb E\{\mathscr E\} = \gStot^\zz(f; T).
\]
\begin{proof}
Let $\Oz = [-1, 1]^{\text{dim}(\zz)}$ where $\text{dim}(\zz)$ denotes the dimension of $\zz$, 
and let $\mu_2$ be the normalized Lebesgue measure on $\Oz$.
We have
\begin{multline*}
   \mathbb E\left\{ \int_0^T \eps(f; t, \zznom)\,dt \right\} 
   = \int_\Oz \int_0^T \eps(f; t, \zznom)\,dt \, \mu_2(d\zznom)
   = \int_0^T \int_\Oz \eps(f; t, \zznom) \, \mu_2(d\zznom) \, dt\\
   = \int_0^T \mathbb E\{ \eps(f; t, \zznom) \} \, dt,
\end{multline*}
where the interchange of integrals follows from Tonelli's theorem. 
Further, by the theorem proved in~\citep{SobolTarantolaGatelliEtAl06},  
$\mathbb E\{ \eps(f; t, \zznom) \} = \Dtot^{U^\complement}(f; t)$
for each $t \in [0, T]$,  
and thus
\[
 \mathbb E\left\{ \int_0^T \eps(f; t, \zznom)\,dt \right\} = \int_0^T \Dtot^{U^\complement}(f; t)\,dt.
\]
This proves the result since $D(f;t)$ is deterministic.
\end{proof}
\end{prop}

A more explicit probabilistic interpretation of Proposition~\ref{prp:basic} can be established by considering the quantity
$\rho = \mathscr E/\gStot^{U^\complement}(f; T)$
and noting that $\mathbb E\{  \rho\} = 1$. As $\rho \ge 0$, the following result is a direct consequence of 
 Markov's inequality.
\begin{cor}
For every $\eps > 0$,
\[
   \prob\Big(\mathscr E \geq \frac{1}{\eps} \gStot^{U^\complement}(f; T)\Big) \leq \eps.
\]
\end{cor}

%
%
\section{Efficient computation of the sensitivity indices}\label{sec:method}
\subsection{Pointwise-in-time surrogate models}\label{sec:pointwise}
To alleviate the cost of computing the generalized
Sobol' indices, we can approximate $f(t, \xx)$ with a 
cheap-to-evaluate surrogate model 
$\tilde{f}(t, \xx)$,  leading to the approximation
\begin{equation}\label{equ:surrogate_sobol}
\gen{S}^U(f; T) \approx \tilde{\gen{S}}^U(f; T) = \frac{\int_0^T D^{U}(\tilde{f}; t)\,dt}{\int_0^T D(\tilde{f}; t)\,dt},
\end{equation}
which can be computed at negligible cost. 
We outline
the corresponding procedure  in Algorithm~\ref{alg:surrogate}. The 
main computational cost  is the evaluations of the process
$f$ at the sampling points $\{\xx^{(j)}\}_{j=1}^N$. Once a surrogate model
$\tilde{f}$ is available, the generalized Sobol' index $\tilde{\gen{S}}^U_\text{tot}(f; T)$
can be computed for any $U \subset \{1, \ldots, \Np\}$ using $\tilde{f}$.

\begin{algorithm}
\renewcommand{\algorithmicrequire}{\textbf{Input:}}
\renewcommand{\algorithmicensure}{\textbf{Output:}}
\caption{Computation of the generalized Sobol' indices via surrogate models 
constructed pointwise in time.}
\label{alg:surrogate}
\begin{algorithmic}[1]
\REQUIRE (i) A quadrature formula on $[0, T]$ with nodes and weights $\{ t_m, w_m \}_{m = 1}^\Nq$.
(ii) function evaluations $\{f(t_m, \xx^{(j)})\}$, $m \in \{1, \ldots, \Nq\}$, 
$j \in \{1, \ldots, N\}$; (iii) An index set $U \subset \{1, \ldots, \Np\}$.

\ENSURE Approximate generalized Sobol' index $\tilde{\gen{S}}^U_\text{tot}(f; T)$. 

\STATE Using the ensemble $\{f(t_m, \xx^{(j)})\}$, construct a surrogate model
$\tilde{f}(t_m, \xx) \approx f(t_m, \xx)$, $m \in \{1, \ldots, \Nq\}$.

\STATE Evaluate the approximate generalized Sobol' index,
\[
\tilde{\gen{S}}^U(f; T) 
= \frac{\sum_{m = 1}^\Nq w_m D^{U}(\tilde{f}; t_m) }{ \sum_{m = 1}^\Nq w_m D(\tilde{f}; t_m)}.
\]
\end{algorithmic}
\end{algorithm}

Polynomial chaos (PC) expansions are commonly used in surrogate modeling.  We now
elaborate on Algorithm~\ref{alg:surrogate} when using PC surrogates. 
PC expansions are series expansion of square integrable random variables in multivariate
orthogonal polynomial bases~\citep{Ghanem:1991a,Xiu10,LeMaitreKnio10}. 
The (truncated) PC representation of $f(t, \xx)$ is of the form 
\begin{equation}\label{equ:PCE}
   f(t, \xx) \approx \sum_{k = 0}^{\Npc} c_k(t) \Psi_k(\xx),
\end{equation}
where $\{\Psi_k\}_{k=0}^\Npc$ is a set of orthogonal polynomials and $\{c_k\}_{k=0}^\Npc$
are expansion coefficients. As $\xx$ 
is assumed to be a $\Np$-dimensional 
uniform random vector, 
we choose $\Np$-variate Legendre polynomials for $\{\Psi_k\}_{k=0}^\Npc$; 
see~\citep{LeMaitreKnio10}. Also, we use total order truncation~\citep{LeMaitreKnio10} and thus
\begin{equation}\label{equ:Npc}
   \Npc + 1 = \frac{\left(\Nord + \Np\right)!}{\Nord!\Np!},
\end{equation}
where $\Nord$ is the maximum total polynomial degree. 
The following are two common approaches for computing PC coefficients 
via sampling; i.e., in a non-intrusive way

\begin{itemize}
\item Non-intrusive spectral projection (NISP), 
\item Regression based methods with sparsity control.
\end{itemize}

Let $u \in L^2_\mu(\Omega) = \{ u : \Omega \to \R : \int_\Omega u(x)^2 \mu(dx) < \infty\}$ to be approximated 
through the  PC representation 
\[
u \approx \sum_{i=0}^\Npc c_k \Psi_k.
\]
The NISP approach~\citep{LeMaitreKnio10,AlexanderianLeMaitrNajmEtAl12,ConradMarzouk13,WinokurConradSrajEtAl13,WinokurKimBisettiEtAl16}, is based on the approximation of Galerkin projections
\begin{equation*}
   \ip{u}{\Psi_l} = 
      \int u(\xx)  \Psi_l(\xx) \mu(d\xx) 
      = \sum_{k = 0}^{\Npc} \int c_k \Psi_k(\xx) \Psi_l(\xx) \, \mu(d\xx)\\
      = \sum_{k = 0}^{\Npc} c_k \ip{\Psi_k}{\Psi_l} 
      = c_l \ip{\Psi_l}{\Psi_l}.
\end{equation*}
through quadrature 
\begin{equation}\label{equ:Gen_Quadrature}
   \ip{u}{\Psi_l} \approx \sum_{j=1}^{\NqNisp} \nu_j
u\big( \xx^{(j)}\big) \Psi_l \big( \xx^{(j)}\big).
\end{equation}
Here $\xx^{(j)} \in \Omega$ and $\nu_j \geq 0$, $j \in\{ 1, \ldots, \NqNisp\}$, 
are quadrature nodes and weights.\footnote{We have denoted the quadrature weights here by $\nu_j$ 
to distinguish them from those in the quadrature formula on the time interval $[0, T]$ 
when computing generalized Sobol' indices; see e.g.,~\eqref{equ:gStot_discrete}.}

Alternatively,  PC coefficients can be computed through regression-based
approaches.  Borrowing ideas from compressive sensing (CS),  sparsity  
is enforced 
by controlling the $\ell_1$
norm of the vector of  PC coefficients. We refer to this  as the CS-based approach.
This approach has been used extensively in recent years for efficient computation 
of PC expansions,  
see e.g.,~\citep{YanGuoXiu12,HamptonDoostan16,FajraouiMarelliSudret17,DiazDoostanHampton18}.

We explain a common formulation of a CS-based approach. We begin by forming 
a sample of points $\{ \xx^{(j)} \}_{j = 1}^N$ in the sample space $\Omega$ and  let $\mathbf{\Lambda} \in
\R^{N\times\Npc}$ be defined by $\Lambda_{jk} = \Psi_k(\xx^{(j)})$, and
$\vec{d} = \big(u(\xx^{(1)}), \ldots, u(\xx^{(N)})\big)^\tran$ be the vector that contains the
function evaluations.  
The vector of PC coefficients is then determined by solving 
\begin{equation}\label{equ:optim}
      \min_{\vec{c} \in \R^{\Npc}} \| \mathbf{\Lambda} \vec{c} - \vec{d} \|_2^2, \quad  
      \text{subject to }  \sum_{k = 0}^\Npc |c_k| \leq \tau. 
\end{equation}
In our computations, we use the solver \textsc{SPGL1}~\citep{spgl1:2007} for the
optimization problem (\ref{equ:optim}).  The parameter $\tau$ that controls
the sparsity of $c$ is found either by trial and error or, more systematically,
through a cross validation procedure. 

Consider now the PC representation $f(t, \xx) \approx \sum_{k = 0}^\Npc c_k(t) \Psi_k(\xx)$.
We have
\begin{equation}\label{equ:Si_PCE}
\gStot^{i}(f; T)
= \frac{\int_0^T \Dtot^{i}(f; t)\,dt}{\int_0^T D(f; t)\,dt}
\approx
\frac{\displaystyle\sum_{k \in \mathcal{K}_i} \| \Psi_k \|^2 \int_0^T c_k(t)^2 \, dt}
{\displaystyle\sum_{k = 1}^\Npc \| \Psi_k \|^2 \int_0^T c_k(t)^2 \, dt},
\quad
i \in \{1, \ldots, \Np\}. 
\end{equation}
Here $\mathcal{K}_i$ is an index set that picks all the terms in the PC
expansion that include $\xi_i$. The definition of this index set is facilitated
by the (partial) tensor product construction of PC basis
functions~\citep{Sudret08,CrestauxLeMaitreMartinez09,Alexanderian13,AlexanderianWinokurSrajEtAl12};
see \ref{sec:Ki} for a brief description.  Note that~\eqref{equ:Si_PCE}
corresponds to $\gStot^U(f; T)$ with $U = \{ i \}$, $i \in \{1, \ldots,
\Np\}$. It is straightforward to generalize the expression for arbitrary $U
\subset \{ 1, \ldots, \Np\}$.

The integrals in~\eqref{equ:Si_PCE} are computed numerically using  
a quadrature formula on $[0, T]$ with nodes and weights
$\{t_m, w_m\}_{m=1}^\Nq$.
This requires computing PC coefficients at every $t_m$, $m \in \{1, \ldots, \Nq\}$. 
The NISP and CS-based approaches for computing PC representation of $f(t, \xx)$ 
share a common feature: a set of function evaluations $f(t_m, \xx^{(j)})$,
$j \in \{ 1, \ldots, N\}$, $m \in \{1, \ldots, \Nq\}$ is needed. This is the main computational
bottleneck for both methods.  

With NISP,  the sampling points are chosen according to a quadrature rule.
The CS-based approach, on the other hand, offers more flexibility and allows for Monte
Carlo or quasi Monte Carlo sampling.  The computational cost of the NISP
numerical quadratures can be very  high, especially with full
tensorization of one-dimensional quadrature rules and/or  when the parameter
dimension is large.  The computational cost can be reduced by carrying out
the integration in~(\ref{equ:Gen_Quadrature}) through Smolyak sparse
quadrature~\citep{Smolyak63,HeissWinschel08}.  A common restriction of both the
NISP and CS-based approaches is the need to access the same set of sampling
points for each $t \in \{t_1, \ldots, t_\Nq\}$. While changing the sampling
points for each time could lead to better approximations, especially if
adaptive quadrature-based approaches are
used~\citep{WinokurConradSrajEtAl13,WinokurKimBisettiEtAl16}, the number of
required function evaluations would be prohibitive.

\begin{algorithm}
\renewcommand{\algorithmicrequire}{\textbf{Input:}}
\renewcommand{\algorithmicensure}{\textbf{Output:}}
\caption{PC-NISP approach for computation of the generalized Sobol' indices}
\label{alg:PCE}
\begin{algorithmic}[1]
\REQUIRE (i) A quadrature formula on $[0, T]$ with nodes and weights $\{ t_m, w_m \}_{m = 1}^\Nq$. 
(ii) a quadrature formula on $\Omega$ with nodes and weights $\{\xx^{(j)}, \nu_j\}_{j=1}^\NqNisp$;
(iii) function evaluations $\{f(t_m, \xx^{(j)})\}$, $m \in \{1, \ldots, \Nq\}$, $j \in \{1, \ldots, \NqNisp\}$;
(iv) a PC basis $\{ \Psi_k \}_{k = 0}^\Npc$.
\ENSURE Approximate generalized total Indices $\tilde{\gen{S}}^i_\text{tot}(f; T)$, $i \in \{1, \ldots, \Np\}$.
\STATE Form the projection matrix
\[
 \Pi_{kj} = \nu_j \Psi_k(\xx^{(j)}) / \ip{\Psi_k}{\Psi_k}, \quad k \in \{0, \ldots, \Npc\}, \,
   j \in \{1, \ldots, \NqNisp\}
\] 
\STATE Compute the vector of PC coefficients at each time step: 
\vspace{-2mm}
\[
\vec{c}(t_m) = \mathbf{\Pi} \vec{d}(t_m), \quad m \in \{1, \ldots, \Nq\}.
\]
\vspace{-2mm}
\STATE Compute approximations to the generalized total sensitivity indices according to~\eqref{equ:Si_PCE}:
\vspace{-2mm}
\[
\tilde{\gen{S}}^{i}_\text{tot}(f; T) =
\frac{\displaystyle\sum_{k \in \mathcal{K}_i} \sum_{m=1}^\Nq \| \Psi_k \|^2 w_m c_k(t_m)^2 }
{\displaystyle\sum_{k = 1}^\Nkl \sum_{m=1}^\Nq \| \Psi_k \|^2 w_m c_k(t_m)^2 }.
\]
\vspace{-2mm}
\end{algorithmic}
\end{algorithm}
To summarize, NISP is a convenient-to-implement approach for computing the
time-dependent PC coefficients and consequently the generalized Sobol' indices,
and can be very effective for certain classes of problems.  We outline the
required steps in Algorithm~\ref{alg:PCE}.

Compared to NISP, CS-based methods present two additional challenges in the above
context: (i) an optimization problem of the
form~\eqref{equ:optim} has to be solved at every $t_m$, $m \in \{1, \ldots,
\Nq\}$ which can be prohibitive when $\Nq$ is large and  (ii) the sparsity
control parameter $\tau$ may need to be calibrated for {\em each}  $t_m$.  CS-based  approaches for the 
joint sparse recovery of function-valued
quantities of interest (such as time-dependent processes) have been proposed; 
they  may provide viable alternatives to a pointwise-in-time
CS-based strategy; see~\citep{DexterTranWebster17} and the references
therein.  

While PC expansions are widely applicable, they present known shortcomings for
certain classes of time-dependent problems; see
e.g.,~\citep{GerritsmaVanderSteenVosEtAl10,AlexanderianLeMaitrNajmEtAl12,
PoetteLucor12, OzenBal17, ChuSudret17} that address problems where
straightforward implementation of a PC-based approach is not optimal.
Depending on the application at hand, other types of surrogates might provide
better alternatives for the purposes of Algorithm~\ref{alg:surrogate}.

\subsection{The spectral approach} \label{sec:spectral}
Even though the approach outlined in Section~\ref{sec:pointwise} can be effective; 
it makes no attempt at exploiting the \emph{structure} of the problem. For instance, and 
as alluded to in the Introduction, computing a KL decomposition  often 
reveals a low-rank representation.  For such cases, the essential features of the corresponding time-dependent processes
are captured with only a few dominant KL modes. Using such a representation,  we can efficiently compute  generalized
Sobol' indices, without the need for surrogate models at every point in time. This is the essence of the  \emph{spectral} 
approach.

Under the notation and assumptions of Section~\ref{sec:var}, 
we represent the process $f$ using the KL expansion
\begin{equation}\label{equ:POD}
   f(t, \xx) = \sum_{i = 1}^\infty f_i(\xx) e_i(t),
\quad 
f_i(\xx) = \int_0^T f(t, \xx) e_i(t) \, dt.
\end{equation}
In practical computations, the above expansion is truncated
\[
\ffn{\Nkl}(t, \xx) = \sum_{i = 1}^\Nkl f_i(\xx) e_i(t), 
\]
with the truncation level $\Nkl$ being informed by the decay of the eigenvalues of $\C$. More precisely, as the variance of the truncated 
KL expansion is given by $\var{\ffn{\Nkl}(t, \xx)} = \sum_{i=1}^\Nkl \lambda_i e_i(t)^2$ 
(cf.\ Lemma~\ref{lem:var}(2) in \ref{sec:proof}) 
it is possible to adjust the truncation level $\Nkl$ by considering the  fraction $r_\Nkl$ of the variance quantified by a given 
truncation level:
\begin{equation}\label{equ:NKL}
   r_\Nkl = \frac{ \int_0^T \var{\ffn{\Nkl}(t, \xx)}\, dt }
                 {\int_0^T \var{f(t, \xx)} \, dt } 
          = \frac{\sum_{i=1}^\Nkl \lambda_i}{\sum_{i=1}^\infty \lambda_i}.
\end{equation}
A similar criterion is used in the computational fluid dynamics community  when truncating
proper orthogonal decompositions (POD) for reduced order modeling \citep{kunisch}. 
The rate at which the eigenvalues of the covariance operator $\C$ decay is problem-dependent. There are, however,
many applications of interest, where the process $f$ corresponds to a
dynamical system with uncertain parameters, for which a small number of KL modes suffice. We call such processes \emph{low-rank}.

\begin{figure}[h]
\centering
\includegraphics[width=1\textwidth]{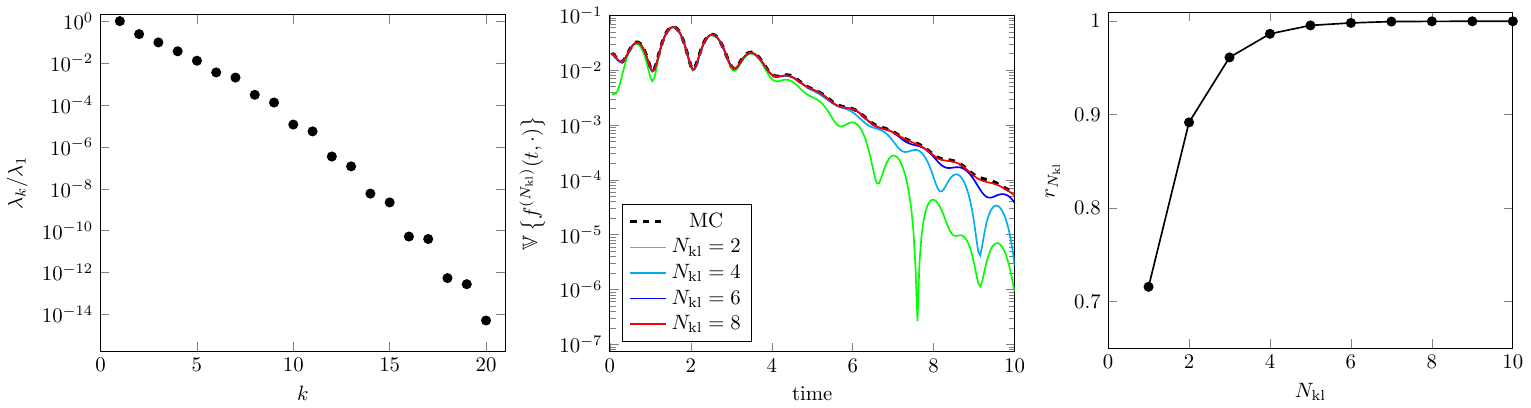}
\caption{Spectral properties of the mechanical oscillator problem \eqref{equ:IVP}. Left: eigenvalues of 
the covariance operator;
middle: pointwise variance of $\ffn{\Nkl}(t, \xx)$ for a few choices of $\Nkl$;  right: the ratio~\eqref{equ:NKL}.}
\label{fig:pendulum_spectrum}
\end{figure}

Figure~\ref{fig:pendulum_spectrum} illustrates the spectral properties of the mechanical oscillator \eqref{equ:IVP}. 
The decay of the first $20$ normalized eigenvalues of the
covariance operator is displayed  in
Figure~\ref{fig:pendulum_spectrum} (left); the rapid decay observed there indicates that a
few KL modes should provide a suitable representation for the process.  For further
insight, we show the evolution of the pointwise variance of $\ffn{\Nkl}(t, \xx)$ for various values of $\Nkl$ (Figure~\ref{fig:pendulum_spectrum}, middle) 
and the behavior of the ratio~\eqref{equ:NKL}
for an increasing number
of KL modes (Figure~\ref{fig:pendulum_spectrum} right).  
The process corresponding to the oscillator problem is an example of 
a low-rank process.
These results are obtained by approximating the covariance function
 according to~\eqref{equ:cov_samp} with a Monte Carlo sample of size
$10^4$.  

Spectral representations can be leveraged to yield efficient algorithms for the
computation of the generalized Sobol' indices. We present the following result that
makes a direct link between the KL expansion of $f$ and the generalized Sobol indices.

\begin{theorem}\label{thm:main}
Let $f$ be a centered process satisfying the assumptions of Section~\ref{sec:var} together with its KL expansion from 
\eqref{equ:POD}.  Then,
for $U \subset \{1, \ldots, \Np\}$,
\begin{equation}\label{equ:spectralSj}
   \gen{S}^U(f; T) =
   \frac{\sum_{i = 1}^\infty  \var{ \E{f_i(\xx) | \xx_U} }}
        {\sum_{i = 1}^\infty \lambda_i},
\end{equation}
where $\lambda_i$ are eigenvalues of the covariance operator $\C$ corresponding 
to the process $f$.
\end{theorem}
\begin{proof}
See \ref{sec:proof}.
\end{proof}

Theorem~\ref{thm:main} yields an efficient approach for numerically
approximating the generalized Sobol' indices in problems where the
eigenvalues of the covariance operator exhibit rapid spectral decay; i.e., for 
low-rank processes. 
In such problems, we can obtain accurate approximations to the generalized sensitivity indices with
only a few modes in the KL expansion. Then, focusing on the expression for
$\gen{S}^U(f; T)$, we consider building surrogate models for the individual
modes $f_i(\xx)$; using these, the variances $\var{ \E{f_i(\xx) | \xx_U}}$ can
be approximated efficiently. 

From the approximate covariance function $c^N(s, t)$ in (\ref{equ:cov_samp}), we construct 
the following approximation 
of the 
covariance operator~\eqref{equ:covariance_operator}: 
\[
    \C^N[u](s) =  
    \int_0^T c^N(s, t) u(t) \, dt, \quad u \in L^2([0, T]). 
\]
This operator 
is then discretized using a quadrature formula in the interval $[0, T]$ with nodes and 
weights $t_m$, $w_m$, $m \in \{ 1, \ldots, \Nq\}$.
To compute the spectral
decomposition of $\C^N$ numerically, 
we have to solve the
discretized (generalized) eigenvalue problem 
\[
   \sum_{m = 1}^\Nq w_m c^N(s_l, t_m) e_i(t_m) = \lambda_i e_i(s_l), 
\quad l \in \{1, \ldots, \Nq\}.
\]
Letting $e_i^l = e_i(t_l)$, $K_{lm} = c^N(s_l, t_m)$, $l, m \in \{1, \ldots, \Nq\}$, and defining the matrix
$\mathbf{W} = \operatorname{diag}(w_1, w_2, \ldots, w_n)$, 
the discretized eigenvalue problem is given by
\[
    \mathbf{K}\mathbf{W} \mathbf{e}_i = \lambda_i \mathbf{e}_i, \quad i \in\{ 1, 2, \ldots, \Nq\}.
\]
This can be rewritten in symmetric form
\begin{eqnarray}
    \mathbf{W}^{1/2} \mathbf{K} \mathbf{W}^{1/2} \mathbf{u}_i = \lambda_i \mathbf{u}_i, \quad i \in\{ 1, 2, \ldots, \Nq\}, \label{equ:evp}
\end{eqnarray}
with $\mathbf{u}_i = \mathbf{W}^{1/2}\mathbf{e}_i$. 
Solving this reformulated eigenvalue problem yields eigenvalues $\lambda_i$ and
eigenvectors $\mathbf{e}_i = \mathbf{W}^{-1/2}\mathbf{u}_i$ that satisfy
\[
   \mathbf{e}_i^\tran \mathbf{W} \mathbf{e}_j = \mathbf{u}_i^\tran 
   \mathbf{W}^{-1/2} \mathbf{W} \mathbf{W}^{-1/2} \mathbf{u}_j = \mathbf{u}_i^\tran \mathbf{u}_j = \delta_{ij},
   \quad i, j \in \{1, \ldots, \Nq\}.
\]
The present approach for computing the eigenvalues and
eigenvectors of the covariance operator is known as the Nystr\"{o}m's method.

\begin{algorithm}[h]
\renewcommand{\algorithmicrequire}{\textbf{Input:}}
\renewcommand{\algorithmicensure}{\textbf{Output:}}
\caption{Spectral-KL approach for computing the generalized total Sobol' indices}
\label{alg:samp}
\begin{spacing}{1}
\begin{algorithmic}[1]
\REQUIRE 
(i) A quadrature formula on $[0, T]$ with nodes and weights $\{ t_m, w_m \}_{m = 1}^\Nq$.
(ii) Function evaluations $\{f(t_m, \xx^k)\}$, $l \in \{1, \ldots, \Nq\}$, 
     $k \in \{1, \ldots, N\}$. 
(iii) An index set $U \subset \{1, \ldots, P\}$.

\ENSURE Generalized Sobol' index $\gen{S}^U(f; T)$.

\STATE Center the process
\vspace{-2mm}
\[
    f_c(t_m, \xx^k) = f(t_m, \xx^k) - \frac1N\sum_{j=1}^N f(t_m, \xx^j),
    \quad k \in \{1, \ldots, N\}, m \in \{1, \ldots, \Nq\}.
\]
\vspace{-2mm}
\STATE Form covariance matrix (discretized covariance function)
\vspace{-2mm}
\[
 K_{lm} =  \frac{1}{N-1} \sum_{k = 1}^{N} f_c(t_l, \xx^k) f_c(t_m, \xx^k), \quad l,m \in \{1, \ldots, \Nq\}.
\]
\vspace{-2mm}
\STATE Let $\mathbf{W} = \mathrm{diag}(w_1, w_2, \ldots, w_\Nq)$ and solve the eigenvalue problem
\vspace{-2mm}
\[
 \mathbf{W}^{1/2} \mathbf{K} \mathbf{W}^{1/2} \mathbf{u}_i = \lambda_i \mathbf{u}_i,
\quad 
 i \in \{ 1, \ldots, \Nq\}. 
\]
\vspace{-4mm}
\STATE Compute $\mathbf{e}_i = \mathbf{W}^{-1/2} \mathbf{u}_i$, $i \in \{ 1, \ldots, \Nq\}$.
\STATE Choose a truncation level $\Nkl$, and compute the discretized KL modes, 
\vspace{-2mm}
\[
{f}_i(\xx^k) = \sum_{m = 1}^\Nq w_m f_c(t_m, \xx^k) e_i^m, \quad i \in \{1, \ldots, \Nkl\},
\, k \in \{1, \ldots, N\}.
\]
\vspace{-2mm}
\STATE Compute a surrogate model for each $f_i$, using function evaluations $\{f_i(\xx^k)\}_{k=1}^N$:
\vspace{-2mm}
\[
   f_i(\xx) \approx \tilde{f}_i(\xx; \xx^1, \xx^2, \ldots, \xx^N).
\]
\vspace{-2mm}
\STATE Compute
\[
\tilde{\gen{S}}^U(f; T) = 
   \frac{\sum_{i = 1}^\Nkl  \var{ \E{\tilde{f}_i(\xx) | \xx_U} }}
        {\sum_{i = 1}^\Nkl \lambda_i}.
\]
\end{algorithmic}
\end{spacing}
\end{algorithm}

Forming the KL expansion requires computing the $f_i$ in~\eqref{equ:POD}; we 
do so via quadrature
\begin{equation}\label{equ:disc_mode}
    {f}_i(\xx) = \sum_{m = 1}^\Nq w_m f(t_m, \xx) e_i(t_m). 
\end{equation}
We can now form the ensemble $\{ f(t_m, \xx^k) \}_{k = 1}^N$ for 
$m \in\{ 1, \ldots, \Nq\}$ and
use it to compute a surrogate model for each mode $f_i$: 
\[
   f_i(\xx) \approx \tilde{f}_i(\xx; \xx^1, \xx^2, \ldots, \xx^N).
\]
This enables efficient approximation of the generalized sensitivity
indices via
\begin{eqnarray}
\gen{S}^U(f; T) \approx \tilde{\gen{S}}^U(f; T) 
:=   \frac{\sum_{i = 1}^\Nkl  \var{ \E{\tilde{f}_i(\xx) | \xx_U} }}
        {\sum_{i = 1}^\Nkl \lambda_i}, \quad U \subset \{1, \ldots, p\}. \label{equ:approxsig}
\end{eqnarray}
The accuracy of the approximation (\ref{equ:approxsig}) depends on (i) the truncation level in the KL expansion, (ii) the accuracy of the 
temporal quadrature in $[0,T]$, (iii) the quality of the sampling in parameter space and (iv)  the error in surrogate model construction.
The various steps of the presented numerical approach for approximating
the generalized sensitivity indices are summarized in Algorithm~\ref{alg:samp}.

\subsection{Implementation of the spectral approach}  \label{sec:computational}
Here, we discuss computational considerations that are
important when implementing Algorithm~\ref{alg:samp}.
\subsubsection{Approximation of the covariance function and the eigenvalue problem}
How sensitive is the discretized eigenvalue problem (\ref{equ:evp}) to the 
number of samples 
$N$ used to construct the approximate covariance function $c^N$? 
We explore this issue numerically.  As an initial test,
Figure~\ref{fig:pendulum_spectrum_conv} (left) displays the first $20$
(normalized) eigenvalues of the approximate covariance operator corresponding to
the oscillator example  \eqref{equ:IVP} as the size of Monte Carlo sample is
varied. We see that even a small Monte Carlo sample (in the order of hundreds)
is sufficient to capture the dominant eigenvalues of $\C$ for this problem.
This is akin to the experiences from the computation of active
subspaces~\citep{Constantine15} where the dominant eigenvalues of a
covariance-like operator are considered.  The impact of using approximate
covariance functions on computation of spectral properties of the covariance
operator is further investigated in Section~\ref{sec:cholera_GSA}.

It is also possible to approximate the covariance function via
quadrature, instead of Monte Carlo,  in the uncertain parameter space.
Performing quadrature is generally challenging for problems with
high-dimensional uncertain parameters; for such problems, full-tensor or
(non-adaptive) sparse grid constructions can be computationally prohibitive 
due to the curse of dimensionality. However, for problems where the
use of a suitable quadrature formula  is feasible, this approach is  preferable
as it yields accurate results. In a quadrature based approach, the sample
average approximations in steps 1 and 2 of Algorithm~\ref{alg:samp} are
replaced by appropriate quadrature formulas. 
We compare the results of computing the eigenvalues of the covariance operator
using a large Monte Carlo sample against a quadrature formula in
Figure~\ref{fig:pendulum_spectrum_conv} (right). 

We also mention that in numerical implementations, the discretized eigenvalue
probelem~\eqref{equ:evp} can be solved efficiently through Krylov iterative
methods. For example we can employ the Lanczos method~\citep{TrefethenBau97} to
compute the dominant eigenpairs of the discretized covariance operator.
Alternatively, one can use randomized methods such as the randomized SVD
algorithm~\citep{HalkoMartinssonTropp11}.

\begin{figure}[h]
\centering
\includegraphics[width=.9\textwidth]{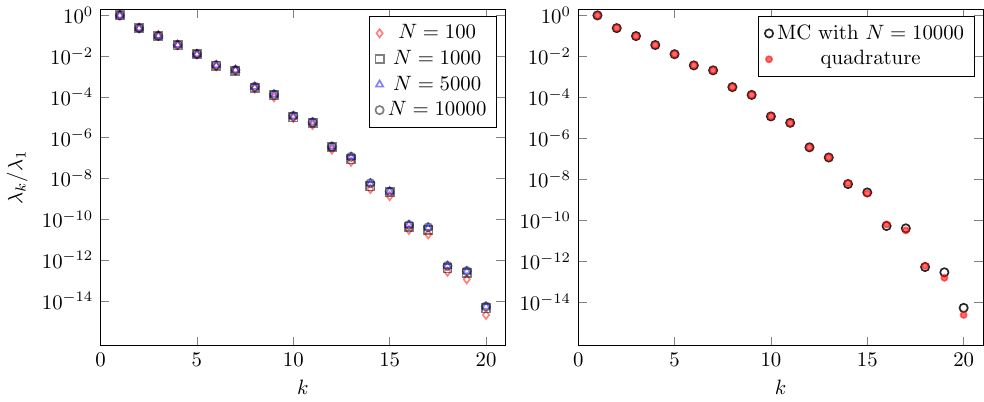}
\caption{Eigenvalues of the discretized covariance operator for the oscillator example~\eqref{equ:IVP}.
Left: influence of the number of Monte Carlo samples $N$ on spectrum; 
right: comparison of the eigenvalues with the covariance function 
approximated via Monte Carlo sampling with 
$10^4$ samples (black circles) and through 
a fully tensorized Gauss-Legendre quadrature in the parameter space with $10^3$
nodes (solid red dots).}
\label{fig:pendulum_spectrum_conv}
\end{figure}

\subsubsection{Polynomial surrogates for the KL modes \eqref{equ:disc_mode}}
Conditional expectations $\E{\tilde{f}_i(\xx) | \xx_U}$ can easily be computed from the 
 PC representation for $f_i(\xx)$
\[
    f_i(\xx) \approx \tilde{f}_i(\xx) =  \sum_{k = 0}^\Npc c_k^i \Psi_k(\xx),
\]
using the tensor product construction of the PC basis, see
e.g.,~\citep{Sudret08}.  This enables efficient computation of 
the generalized Sobol' indices.
For example, 
$\tilde{\gen{S}}^U(f; T)$ with $U = \{ j \}$, $j\in \{1, \ldots, \Np\}$ 
can be approximated as follows:
\[
\gen{S}^j(f; T) 
\approx
\tilde{\gen{S}}^j(f; T) =  
   \frac{\sum_{i = 1}^\Nkl  \sum_{k \in \mathcal{I}_j} \| \Psi_k \|^2 (c_k^i)^2 }
        {\sum_{i = 1}^\Nkl \lambda_i}, \quad j \in \{1, \ldots, \Np\}.
\]
where $\mathcal{I}_j$ 
is an index set that picks all the terms in the PC
expansion that include only $\xi_i$.
The computation of the PC expansion coefficients for $f_i$ themselves can be
done through a CS-based approach or NISP as outlined earlier. 
Note also that the generalized total Sobol' indices can be approximated through 
$\tilde{\gen{S}}^j_\text{tot}(f; T) = 1 - \tilde{\gen{S}}^{U_{\sim j}}(f; T)$ 
with $U_{\sim j} = \{1, \ldots, \Npc\} \setminus \{ j \}$.

In Figure~\ref{fig:pendulum_barplot_compare}, we compare the performance of
several options within Algorithm~\ref{alg:samp}.  We present results using a
quadrature based approach, where  we approximate the covariance function via
quadrature, and compute PC representations for $f_i$ via NISP; specifically, we
consider a full-tensor quadrature formula and a Smolyak sparse
quadrature formula (see the Figure caption for more details). 
We also use Algorithm~\ref{alg:samp} with a small Monte Carlo
sample in the uncertain parameter space, where we compute the covariance
function via sample averaging, and compute the PC representations of $f_i$
using the CS-based approach.  Moreover, we report the generalized sensitivity
indices computed using direct Monte Carlo sampling with $10^5$ samples. Results from
all approaches agree remarkably well.

\begin{figure}[h]
\centering
\includegraphics[width=.5\textwidth]{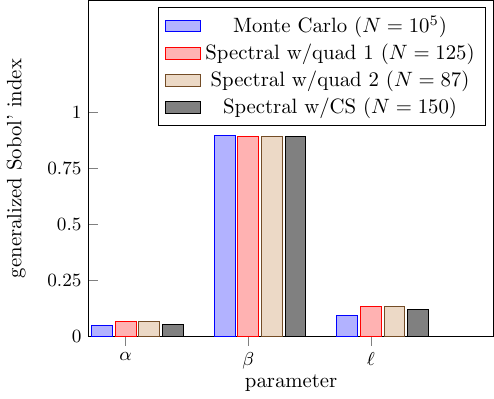}
\caption{Generalized Sobol' indices computed via Algorithm~\ref{alg:samp} for the mechanical oscillator  example~\eqref{equ:IVP};  {\em spectral w/quad 1:}  full tensor Gauss-Legendre
quadrature with five nodes in each dimension, {\em spectral w/quad 2:}  Smolyak sparse grid  based on delayed Kronrod--Patterson 
rule~\citep{petras:2003,HeissWinschel08} and {\em spectral CS:} Monte Carlo sample of size $150$. In each case, 
a fourth order PC expansion for $f_i$, $i = 1, \ldots, \Nkl$, is computed ($\Nkl = 8$).}
\label{fig:pendulum_barplot_compare}
\end{figure}

It is important to note that Algorithm~\ref{alg:samp} utilizes  
the surrogate models $\tilde{f}_i$ to approximate the conditional 
expectations in the numerator of~\eqref{equ:spectralSj}, as a means
of approximating the generalized Sobol' indices. This is not the same as
computing the exact 
generalized Sobol' indices of the approximate KL expansion, 
$\tilde{f}(t, \xx) = \tilde{f}_0(t) + \sum_{i=1}^\Nkl \tilde{f}_i(\xx) e_i(t)$.
Another alternative approach for computing the generalized 
Sobol' indices is provided by sampling the approximate truncated KL 
expansion of $f$, as explained below. 

\subsubsection{The approximate KL expansion as a global surrogate model}
The computations performed in Algorithm~\ref{alg:samp} lead to 
an approximate KL representation of $f$,
\begin{equation}\label{equ:global}
    f(t_m, \xx) \approx \tilde{f}(t_m, \xx) := 
   \tilde{f}_0(t_m) + \sum_{i=1}^\Nkl \tilde{f}_i(\xx) e_i^m, \quad 
    m \in \{1, \ldots, \Nq\}, 
\end{equation}
where $e_i^m$ is as in Algorithm~\ref{alg:samp}, and $\tilde{f}_0(t_m)$ is the
sample mean, at $t = t_m$, computed in Algorithm~\ref{alg:samp}. This provides a
cheap-to-evaluate surrogate model, which can be used for an alternative
approach of approximating generalized Sobol' indices via sampling $\tilde{f}(t_m, \xx)$. 
The utility of this surrogate model, however,
extends beyond sensitivity analysis: $\tilde{f}(t, \xx)$ can be used to accelerate various
uncertainty quantification tasks, where repeated evaluations of $f(t, \xx)$ are
required.
We point out that a related approach was implemented 
in~\citep{LiIskandaraniLeHenaffEtAl16} for representation of spatially 
distributed processes.

%
%
\section{Probabilistic modeling and sensitivity analysis for a cholera model}\label{sec:cholera}
\def\ve{{\varepsilon}}
\def\ds{\displaystyle}

We illustrate attributes of the generalized sensitivity indices in the
context of a cholera model proposed in \citep{Hartley}.  

\subsection{Model description}
A population of $N_\text{pop}$ subjects is split, at time $t$, into $S(t)$
susceptible individuals, $I(t)$ infectious individuals, and $R(t)$ recovered
individuals; the model assumes the total population $N_\text{pop}$ to stay
constant while $S$, $I$ and $R$ vary during an epidemic with $N_\text{pop} =
S(t) + I(t)+R(t)$.  Also  considered are concentrations $B_H(t)$ and $B_L(t)$
of highly- and lowly-infectious cholera bacteria, {\em Vibrio cholerae}.  The
units for these five state variables are compiled in Table~\ref{table1}. We
illustrate the associated compartment model in Figure~\ref{fig:SIR_comp}.

\vspace{2mm}
\begin{minipage}{\textwidth}
\begin{minipage}[b]{0.32\textwidth}
\centering
\includegraphics[width=.7\textwidth]{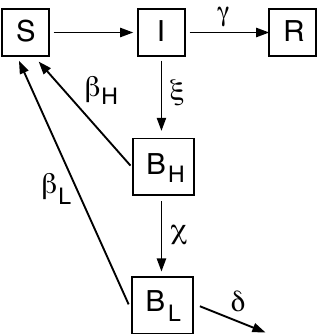}
\captionof{figure}{Compartmental cholera model 
from~\citep{Hartley}.} 
\label{fig:SIR_comp}
\end{minipage}
\begin{minipage}[b]{0.65\textwidth}
\centering
\begin{tabular}{l|c|c} \hline
State & Symbol & Units \\ \hline
Susceptible Individuals & $S$ & \# individuals \\
Infected Individuals & $I$ &  \# individuals \\
Recovered Individuals & $R$ &  \# individuals \\
Concentration of highly-infectious  & $ B_H$ & 
   $\frac{\mbox{\rm \# bacteria}}{m \ell}$ \\
\quad cholera bacteria &&\\
Concentration of lowly-infectious & $ B_L$ & $\frac{\mbox{\rm \# bacteria}}{m \ell}$ \\
\quad cholera bacteria &&\\\hline
\end{tabular}
\captionof{table}{State variables and units for the cholera model.}
\label{table1}
\end{minipage}
\end{minipage}
\vspace{2mm}

The cholera model in \citep{Hartley}
is based on the following assumptions. 
(i) The birth and death rates are identical and denoted by $b$.
(ii) Susceptible individuals become infected by drinking bacteria-infested
water. The rate at which this happens is proportional to $S(t)$, the
concentrations $B_H$ and $B_L$ of highly and lowly-infectious bacteria, and the
drinking rates $\beta_H$ and $\beta_L$ at which these bacteria are ingested.
The rates also satisfy the saturation relations that when $B_H = \kappa_H$ and
$B_L = \kappa_L$, where $\kappa_H$ and $\kappa_L$ denote cholera carrying
capacities, the probability of ingestion resulting in disease is 0.5.
Susceptibles recover at a rate $\gamma$. 
(iii)  Infected individuals spread highly-infectious bacteria $B_H$ to the water
at a rate $\zeta$.
(iv) Highly-infectious bacteria $B_H$ become lowly infectious
$B_L$ at a rate $\chi$.
(v) Lowly-infectious bacteria $B_L$ die at a rate
$\delta$.

These assumptions yield the system of ordinary differential equations (ODEs) 
\begin{equation}
\begin{array}{l}
  {\ds \frac{dS}{dt} = bN_\text{pop} - \beta_L S \frac{B_L}{\kappa_L + B_L} - \beta_H S \frac{B_H}{\kappa_H + B_H}  - b S} \\
  \noalign{\medskip}
  {\ds \frac{dI}{dt} =  \beta_L S \frac{B_L}{\kappa_L + B_L} + \beta_H S \frac{B_H}{\kappa_H + B_H} - (\gamma + b)I} \\
  \noalign{\medskip}
  {\ds \frac{dR}{dt} = \gamma I - b R} \\
  \noalign{\medskip}
  {\ds \frac{d B_H}{d t} = \zeta I - \chi B_H} \\
  \noalign{\medskip}
  {\ds \frac{d B_L}{dt} = \chi B_H - \delta B_L}
\end{array}
\label{eq1}
\end{equation}
with initial conditions 
$(S(0), I(0), R(0), B_H(0), B_L(0)) = (S_0, I_0, R_0, B_{H_0},  B_{L_0})$.

The parameter units and nominal values from \citep{Hartley} are compiled in
Table~\ref{table2}.  We note that $\frac{d S}{d t} + \frac{d I}{d t} +\frac{d
R}{d t} = 0$ so that $S(t) + I(t) + R(t) = N_\text{pop}$ and the population size indeed remains
constant.  The system dynamics are illustrated  in Figure~\ref{fig:SIR_model}.

Our simulations correspond to a total population of $N_\text{pop}=10{,}000$ with
initial states given by  $S_0 = N_\text{pop} - 1$, $I_0 = 1$, $R_0 = 0$, and $B_{H_0} =
B_{L_0} = 0$. We solve the problem up to time  $T = 250$.
The ODE system is integrated using the solver \verb+ode45+ provided in
\textsc{Matlab} ODE toolbox. We use absolute and relative tolerances of $10^{-6}$ for
the ODE solver.  The solution is recorded at 
$t_i = i \Delta t$, $i \in \{0, \ldots, \Nq\}$, with $\Delta t = 5\times10^{-2}$ and $\Nq
= 250 / \Delta t$. The temporal integrals  from
Algorithms~\ref{alg:PCE}~and~\ref{alg:samp} are evaluated through the composite trapezoidal
rule, where the quadrature nodes are the time steps $\{ t_i \}_{i = 0}^\Nq$.

\renewcommand{\arraystretch}{1.1}
\begin{table}[!tt]
\begin{center}
\begin{tabular}{l|c|c|c} \hline
Model Parameter & Symbol & Units & Values \\ \hline
Rate of drinking $B_L$ cholera & $\beta_L$ & $\frac{1}{\mbox{\rm week}}$ & 1.5 \\
Rate of drinking $B_H$ cholera & $\beta_H$ & $\frac{1}{\mbox{\rm week}}$  & 7.5~($\ast$) \\
$B_L$ cholera carrying capacity & $\kappa_L$ & $\frac{\mbox{\rm \# bacteria}}{m\ell}$ & $10^6$  \\
$B_H$ cholera carrying capacity & $\kappa_H$ & $\frac{\mbox{\rm \# bacteria}}{m\ell}$ & $\frac{\kappa_L}{700}$ \\
Human birth and death rate & $b$ & $\frac{1}{\mbox{\rm week}}$  & $\frac{1}{1560}$ \\
Rate of decay from $B_H$ to $B_L$ & $\chi$ & $\frac{1}{\mbox{\rm week}}$ & $\frac{168}{5}$ \\
Rate at which infectious individuals & $\zeta$  & $\frac{\mbox{\rm \# bacteria}}{\mbox{\rm \# individuals} \cdot m \ell \cdot \mbox{\rm week}}$ & 70 \\ \noalign{\vspace{-.1truein}}
spread $B_H$ bacteria to water & & & \\
Death rate of $B_L$ cholera & $\delta$ & $\frac{1}{\mbox{\rm week}}$ & $\frac{7}{30}$ \\
Rate of recovery from cholera & $\gamma$ & $\frac{1}{\mbox{\rm week}}$ & $\frac{7}{5}$ \\ \hline
\end{tabular}
\end{center}
\renewcommand{\arraystretch}{1.0}
\caption{Cholera model parameters from \citep{Hartley}. \\ 
($\ast$) The value  $\beta_H = 7.5$  is consistent with \citep{Hartley} where it is assumed that $\beta_H > \beta_L$; no corresponding nominal value for $\beta_H$ was, however, provided there. }
\label{table2}
\end{table}

\begin{figure}[ht]
\centering
\includegraphics[width=.8\textwidth]{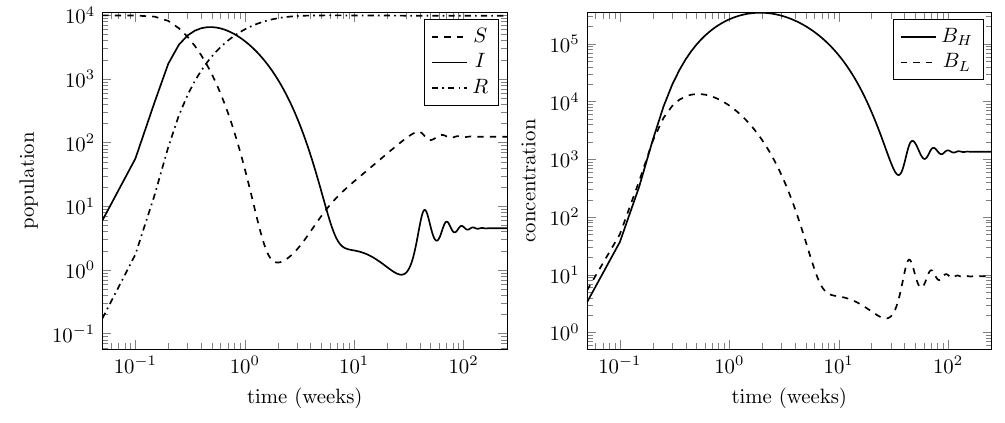}
\caption{
Cholera model (\ref{eq1}):  time evolution of $S$, $I$, $R$, $B_H$, and $B_L$.}
\label{fig:SIR_model}
\end{figure}

\subsection{Statistical model for uncertain model parameters and the quantity of interest}
The parameter vector $\vec{x} = (\beta_L, \beta_H, \kappa_L, b, \chi, \zeta,
\delta, \gamma)$ is considered as uncertain. The nominal values $\bar{\vec{x}}$ for these
parameters are specified in Table~\ref{table2}.  The distribution of
these uncertain parameters is taken as uniform, with $10\%$ perturbation around
the respective nominal values:
\begin{equation}\label{equ:cholera_parameterization}
   x_i = \bar{x}_i + 0.1 \bar{x}_i \xi_i, \quad \xi_i \sim \mathcal U(-1,1), \quad
   i \in \{1, \ldots, \Np\}, \quad
    \Np = 8.
\end{equation}
In~\citep{Hartley}, the nominal value of $\kappa_H$ is taken as 
${\kappa_L}/{700}$.  Hence, we set
 $\kappa_H = x_3 / 700$, where $x_3$ is as
in~\eqref{equ:cholera_parameterization}.

Our quantity of interest is the infected population $I$ as a function
of time.  Since the vector $\vec{x}$ of the uncertain model parameters is
defined by the random vector $\xx$ in~\eqref{equ:cholera_parameterization}, we
can consider the infected population as a process $I(t, \xx)$.

\subsection{Global sensitivity analysis}\label{sec:cholera_GSA} 
The traditional total Sobol' indices, computed pointwise in time, are displayed
in Figure~\ref{fig:all_indices} (left).  These indices, which show great
variation over time, are difficult to interpret.  This is mainly due to fact
that the variance of the quantity of interest, infected population, itself
varies significantly over time; see Figure~\ref{fig:pointwise_var}.
\begin{figure}[ht]\centering
\includegraphics[width=1\textwidth]{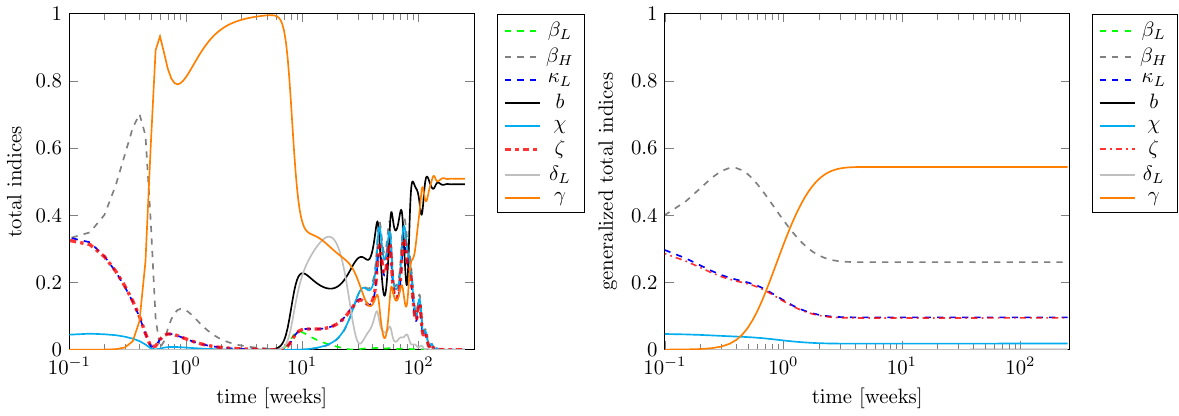} 
\caption{Sensitivity analysis of the cholera model (\ref{eq1}). Left: pointwise in time total sobol indices; right:
generalized total indices over successively larger time-intervals.}
\label{fig:all_indices}
\end{figure}

\begin{figure}[ht]\centering
\includegraphics[width=0.4\textwidth]{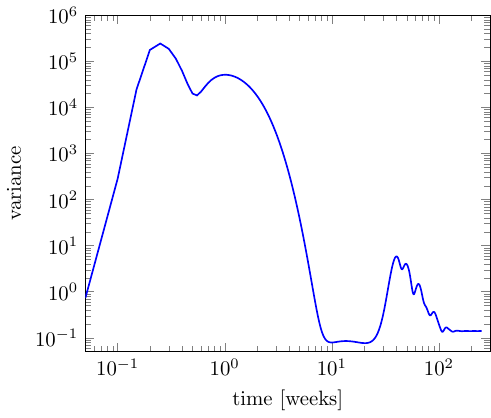}
\caption{Pointwise variance over time.}
\label{fig:pointwise_var}
\end{figure}

By contrast,  the generalized total Sobol' indices offer a more robust
and meaningful picture; see  Figure~\ref{fig:all_indices}~(right), where generalized total
indices are computed over successively larger time intervals and
Figure~\ref{fig:cholera_sensitivity} that reports the generalized Sobol' indices
corresponding to the entire simulation time interval.   
Note, for example, the behavior of the variable $b$: based on pointwise-in-time
Sobol' indices in Figure~\ref{fig:all_indices}~(left), one might make the conclusion 
that $b$ is an important variable. However, with Figure~\ref{fig:pointwise_var} in mind, 
we note that $b$ becomes a major 
contributor to output variance around the end of the simulation time, when the
model variance is at most $\mathcal{O}(1)$. In contrast, contribution of $b$
to model variance is negeligible 
earlier in simulation time when the model variance is $\mathcal{O}(10^5)$. 
The generalized total Sobol' indices, which are aware of the history of the
process, incorporate this and accordingly rank $b$ is an unimportant variable.

\begin{figure}[ht]\centering
\includegraphics[width=1\textwidth]{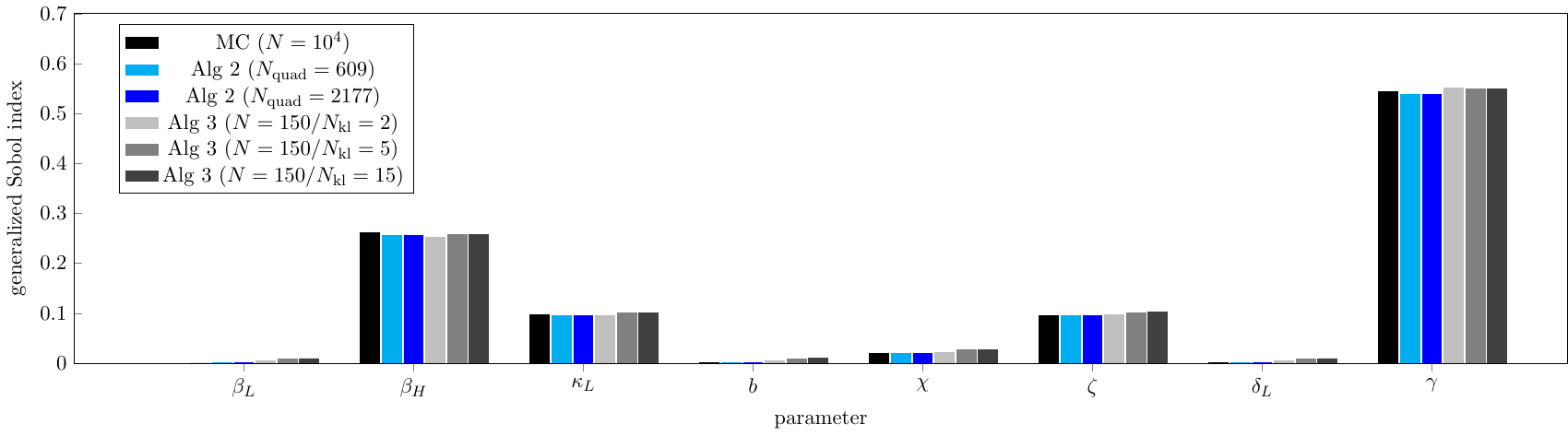}
\caption{Generalized total Sobol' indices for the cholera model (\ref{eq1}) with $T = 250$, computed six different ways.}
\label{fig:cholera_sensitivity}
\end{figure}

While the generalized
indices are computed six different ways, as we now detail, a strong consistency can
be observed in the result. 
When computing the generalized Sobol' indices via
Algorithm~\ref{alg:PCE}, we use a third-order PC expansion of $I(t, \xx)$,
for each $t$ in the time grid. The PC coefficients are computed via NISP with
sparse quadrature of varying resolutions.  The  results in
Figure~\ref{fig:cholera_sensitivity}  indicate
that with approximately six hundred model evaluations, it is possible to construct a PC surrogate
model that is  suitable for accurate estimation of the generalized Sobol' indices. 

For Algorithm~\ref{alg:samp}, we use Monte Carlo sampling to approximate the covariance function and rely on the
CS-based approach from Section~\ref{sec:pointwise} to approximate the
fourth-order PC coefficients of the KL modes $f_i$ of the random process. 
The approach only requires a small number of model evaluations
to accurately estimate the Sobol' indices for this problem. 

A key component of 
Algorithm~\ref{alg:samp} is the spectral decomposition of the covariance operator
of the process.  
 Figure~\ref{fig:spectrum_var} (left) displays the normalized eigenvalues of
the covariance operator as the Monte Carlo sample size increases. We
note both a rapid spectral decay and the fact that a small number of Monte Carlo
samples is sufficient to accurately estimate the dominant eigenvalues.  
As shown in Figure~\ref{fig:spectrum_var} (middle), the evolution of the pointwise covariance, computed using a truncated KL expansion,
 can be quantified accurately with a small number of KL modes. Indeed, $\Nkl = 15$ modes is sufficient; in fact, 
only  two KL
modes can be used in the interval $[0, 1]$; i.e., during the transient regime.  

The computation in Figure~\ref{fig:spectrum_var} (middle) uses a fixed Monte
Carlo sample of size $N = 10^4$. From the results in
Figure~\ref{fig:spectrum_var} (left), we know that a much smaller Monte Carlo 
sample is sufficient for approximating the dominant eigenvalues. However, the 
computation of the pointwise variance depends also on approximation of the
eigenvectors of the covariance operator. Instead of performing convergence
studies for the dominant eigenvectors, we consider the following question:  how does the
approximation of the pointwise variance change for $\Nkl = 15$ if we use
smaller Monte Carlo samples? This is investigated in
Figure~\ref{fig:spectrum_var} (right) which confirms that a small Monte Carlo sample
enables accurate estimation of the pointwise variance, as computed by 
an approximate truncated KL expansion. 
\begin{figure}[h]
\includegraphics[width=1\textwidth]{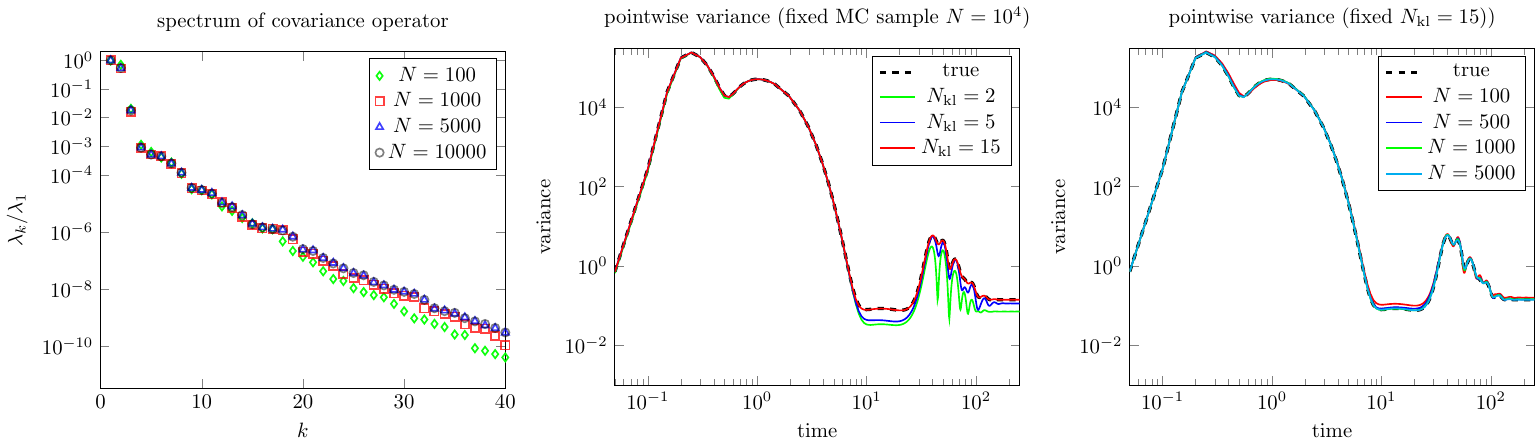}
\caption{Spectral properties and variance analysis of the cholera model (\ref{eq1}). Left: spectrum of the sampled covariance operator from Section~\ref{sec:spectral};  middle: 
pointwise variance of the truncated process $\ffn{\Nkl}(t, \xx)$ with $N = 10^4$ samples to 
approximate the covariance function.
Right: pointwise variance of the truncated process $\ffn{\Nkl}(t, \xx)$
with $\Nkl = 15$ and with varying Monte Carlo sample sizes used to approximate the covariance
function. }
\label{fig:spectrum_var}
\end{figure}

\subsection{Generalized Sobol' indices for parameter dimension reduction}
Based on generalized Sobol' indices on the interval $[0, 250]$, the important
variables are $\beta_H$, $\kappa_L$, $\zeta$, and $\gamma$.  This suggests that we can reduce the parameter dimension by fixing the remaining variables at their
nominal values. We provide next a numerical study of the approximation
errors which result from fixing inessential variables. To illustrate the potential pitfalls of fixing parameters based
on pointwise in time classical Sobol' indices, we also consider fixing parameters
according to the classical Sobol' indices at $t = 250$, which indicate $b$ and
$\gamma$ as the important parameters.

When fixing inessential variables, we consider the reduced model
\[
   \tilde{I}_u(t, \xx) = I(t, \tilde\xx), \quad \text{with }\tilde\xx 
      = (\xx_U, \xx_{U^\complement}^\text{nom}), 
   \quad U \subseteq X = \{1, \ldots, \Np\}.
\]
To provide a thorough study, we examine the impact of fixing 
parameters on pointwise variance and the distribution of the process.

\begin{figure}[ht]\centering
\includegraphics[width=0.32\textwidth]{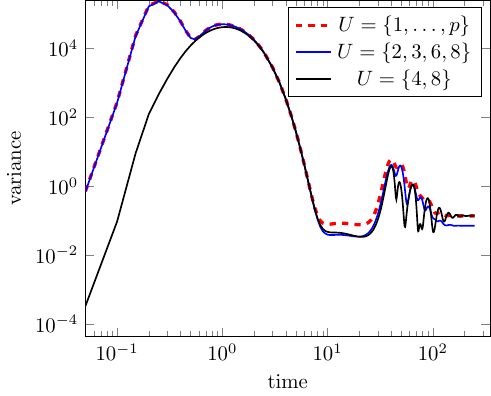}
\includegraphics[width=0.32\textwidth]{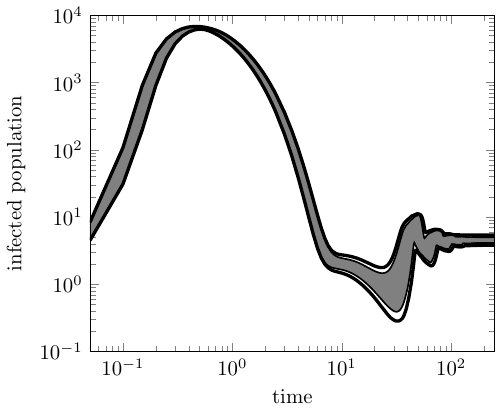}
\includegraphics[width=0.32\textwidth]{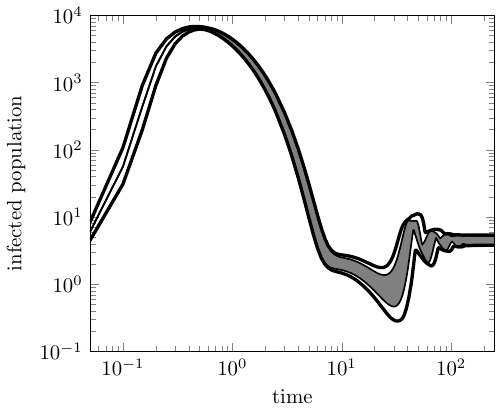}
\caption{Effect of fixing inessential variables for the cholera model (\ref{eq1}).
Left: effect on the variance;
middle: effect on the quantity of interest when using generalized indices; 
right: effect on the quantity of interest when using  pointwise
indices at the final time. }
\label{fig:fixing_var}
\end{figure}

Figure~\ref{fig:fixing_var} (left) illustrates the impact of fixing inessential variables on the
variance of the process over time.  The effect  on the distribution of the process itself is studied in
Figures~\ref{fig:fixing_var}~(middle) and (right), where
important variables are chosen based on generalized Sobol' indices and  based on
pointwise Sobol' indices at the final time, respectively.  The thick black
lines  indicate the
$2$nd and $98$th percentiles obtained by sampling $I(t, \xx)$ (with no
variables fixed) $10^4$ times. The shaded regions enclose the respective
percentiles for $\tilde{I}(t, \xx)$, which is also obtained by sampling the
reduced models $10^4$ times. 

We note that fixing variables according to generalized indices results in a
reduced model that captures the distribution of $I$ over the simulation time
window well. On the other hand, and as expected, fixing variables according to
pointwise Sobol' indices at the final time is effective at capturing the
distribution of $I$ only as the system approaches equilibrium.
In Figure~\ref{fig:fixing_var_dist_composite},
we study the impact of fixing variables on the probability density
function (PDF) of the infected population over time. These
PDFs were generated by sampling the reduced models $10^4$ times.
\begin{figure}[ht!]
\includegraphics[width=0.95\textwidth]{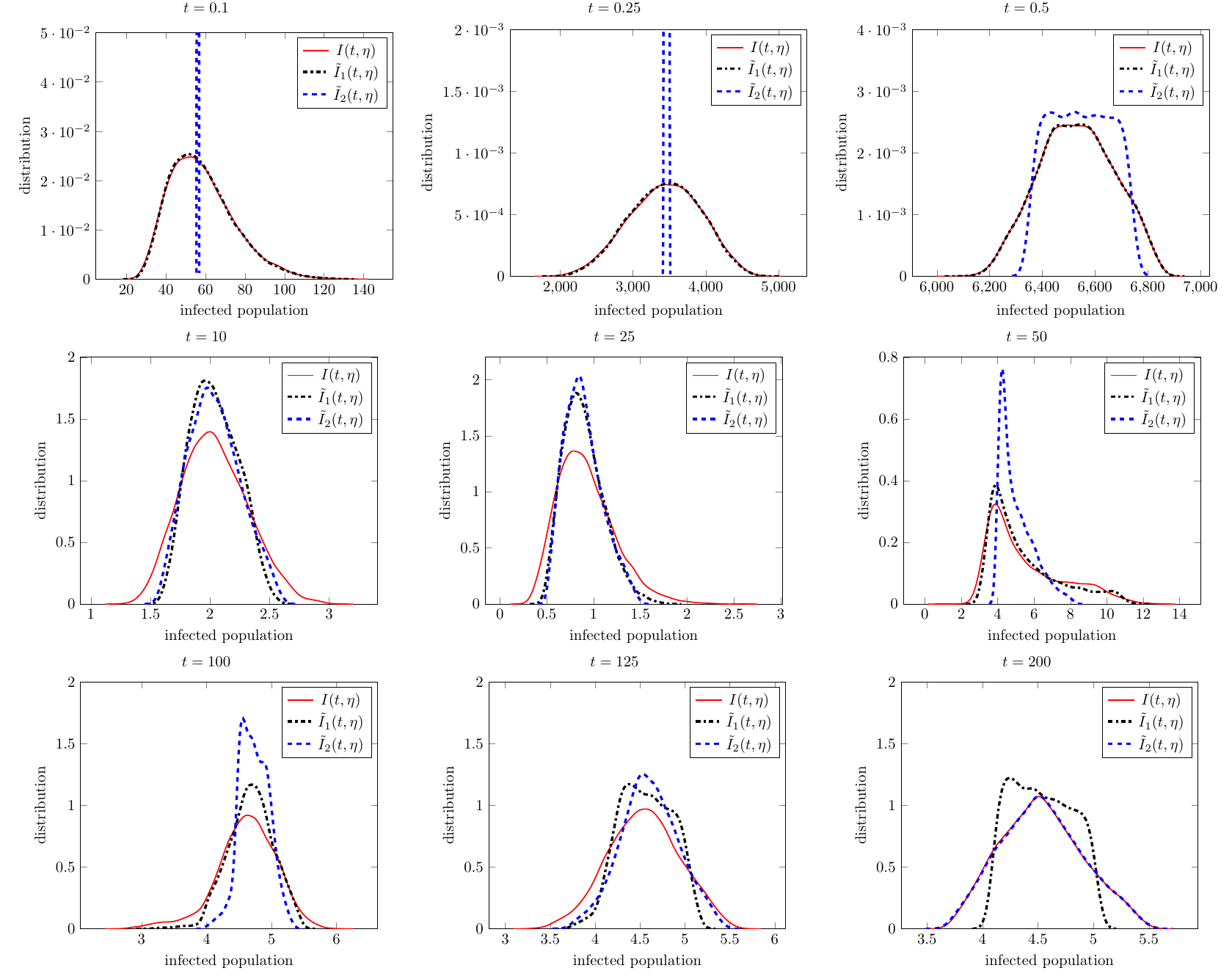}
\caption{Effect of fixing inessential variables, 
chosen according to generalized indices (indicated as $\tilde{I}_1$)
and Sobol' indices at final time (indicated as $\tilde{I}_2$)
on distribution of the 
infected population for the cholera model (\ref{eq1}).}
\label{fig:fixing_var_dist_composite}
\end{figure}

\section{Conclusions}
The global sensitivity analysis of time-dependent processes such as
(\ref{equ:basic_model}) requires history-aware approaches. Not surprisingly,
identifying inessential parameters  based on a pointwise in time analysis, such
as the one corresponding to the standard Sobol' indices, is only valid close to
the time at which the analysis is performed. For applications where the {\em
evolution} of the process under study is of interest, sensitivity analysis must
be performed globally in time. 

We show how to efficiently compute generalized Sobol' indices.
The various tests of the impact of fixing inessential parameters provide a
consistent picture: using generalized total Sobol' indices, we can reliably
select the variables with dominant impact on variability of the quantity of
interest. Further formalization of these results in the form of, for instance,
theoretical error analysis is to our knowledge not available but is desirable.
Likewise, the global sensitivity analysis of  time-dependent processes with
correlated parameters is beyond this work and deserves further investigation.

\section*{Acknowledgements} 
The research of PAG was supported in part by the National Science Foundation
through grants DMS-1522765 and DMS-1745654. The research of RCS was supported
in part by the Air Force Office of Scientific Research (AFOSR) through grant
AFOSR~FA9550-15-1-0299.

\bibliographystyle{elsarticle-num}
\bibliography{refs}

\begin{thebibliography}{10}
\expandafter\ifx\csname url\endcsname\relax
  \def\url#1{\texttt{#1}}\fi
\expandafter\ifx\csname urlprefix\endcsname\relax\def\urlprefix{URL }\fi
\expandafter\ifx\csname href\endcsname\relax
  \def\href#1#2{#2} \def\path#1{#1}\fi

\bibitem{Sobol93}
I.~{Sobol'}, Sensitivity estimates for nonlinear mathematical models, Math.
  Modeling Comput. Experiment 1 (1993) 407--414.

\bibitem{Sobol01}
I.~Sobol', Global sensitivity indices for nonlinear mathematical models and
  their {Monte Carlo} estimates, Mathematics and Computers in Simulation 55
  (2001) 271--280.

\bibitem{Owen13}
A.~{Owen}, Better estimation of small {Sobol'} sensitivity indices, ACM
  Transactions on Modeling and Computer Simulation 23 (2013) Art. 11, 17.

\bibitem{SaltelliRattoAndresEtAl08}
A.~{Saltelli}, M.~{Ratto}, T.~{Andres}, F.~{Campolongo}, J.~{Cariboni},
  D.~{Gatelli}, M.~{Saisana}, S.~{Tarantola}, Global sensitivity analysis: the
  primer, Wiley, 2008.

\bibitem{AlexanderianWinokurSrajEtAl12}
A.~Alexanderian, J.~Winokur, I.~Sraj, A.~Srinivasan, M.~Iskandarani, W.~C.
  Thacker, O.~M. Knio, Global sensitivity analysis in an ocean general
  circulation model: a sparse spectral projection approach, Computational
  Geosciences 16~(3) (2012) 757--778.

\bibitem{Namhata2016OladyshkinDilmoreEtAl16}
A.~Namhata, S.~Oladyshkin, R.~M. Dilmore, L.~Zhang, D.~V. Nakles, Probabilistic
  assessment of above zone pressure predictions at a geologic carbon storage
  site, Scientific reports 6 (2016) 39536.

\bibitem{GamboaJanonKleinEtAl14}
F.~Gamboa, A.~Janon, T.~Klein, A.~Lagnoux, Sensitivity analysis for
  multidimensional and functional outputs, Electronic Journal of Statistics
  8~(1) (2014) 575--603.

\bibitem{LamboniMonodMakowski11}
M.~Lamboni, H.~Monod, D.~Makowski, Multivariate sensitivity analysis to measure
  global contribution of input factors in dynamic models, Reliability
  Engineering \& System Safety 96~(4) (2011) 450--459.

\bibitem{Wiener:1938}
N.~Wiener, The {Homogeneous Chaos}, American Journal of Mathematics 60 (1938)
  897--936.

\bibitem{LeMaitreKnio10}
O.~{Le Ma{\^\i}tre}, O.~Knio, Spectral Methods for Uncertainty Quantification
  With Applications to Computational Fluid Dynamics, Scientific Computation,
  Springer, 2010.

\bibitem{Xiu10}
D.~Xiu, Numerical methods for stochastic computations: a spectral method
  approach, Princeton University Press, 2010.

\bibitem{friedman93}
J.~{Friedman}, Fast {MARS}, Tech. Rep. 110, Laboratory for Computational
  Statistics, Department of Statistics, Stanford University (1993).

\bibitem{CrestauxLeMaitreMartinez09}
T.~Crestaux, O.~L. Maitre, J.-M. Martinez, Polynomial chaos expansion for
  sensitivity analysis, Reliability Engineering \& System Safety 94~(7) (2009)
  1161--1172, special Issue on Sensitivity Analysis.

\bibitem{Sudret08}
B.~Sudret, Global sensitivity analysis using polynomial chaos expansions,
  Reliability Engineering \& System Safety 93~(7) (2008) 964--979.

\bibitem{BlatmanSudret10}
G.~Blatman, B.~Sudret, Efficient computation of global sensitivity indices
  using sparse polynomial chaos expansions, Reliability Engineering \& System
  Safety 95~(11) (2010) 1216--1229.

\bibitem{Alexanderian13}
A.~Alexandrian, On spectral methods for variance based sensitivity analysis,
  Probability Surveys 10 (2013) 51--68.

\bibitem{HartAlexanderianGremaud17}
J.~Hart, A.~Alexanderian, P.~Gremaud, Efficient computation of {Sobol'} indices
  for stochastic models, SIAM Journal on Scientific Computing 39 (2017)
  A1514--A1530.

\bibitem{kleijnen}
J.~{Kleijnen}, W.~{van Beers}, Kriging for interpolation in random simulations,
  Journal of the Operational Research Society 54 (2003) 255--262.

\bibitem{legratiet}
L.~{Le Gratiet}, C.~{Cannamela}, B.~{Iooss}, A bayesian approach for global
  sensitivity analysis of (multifidelity) computer codes, SIAM/ASA Journal on
  Uncertainty Quantification 2 (2014) 336--363.

\bibitem{CampbellMcKayWilliams06}
K.~Campbell, M.~D. McKay, B.~J. Williams, Sensitivity analysis when model
  outputs are functions, Reliability Engineering \& System Safety 91~(10-11)
  (2006) 1468--1472.

\bibitem{HongLuyi16}
H.~Xiao, L.~Li,
  \href{http://www.sciencedirect.com/science/article/pii/S0951832015002975}{Discussion
  of paper by {Matieyendou Lamboni, Herv\'{e} Monod, David Makowski}
  ``{M}ultivariate sensitivity analysis to measure global contribution of input
  factors in dynamic models'', {R}eliab. {E}ng. {S}yst. {S}af. 99 (2011)
  450--459}, Reliability Engineering \& System Safety 147 (2016) 194 -- 195.
\newblock \href {https://doi.org/https://doi.org/10.1016/j.ress.2015.10.015}
  {\path{doi:https://doi.org/10.1016/j.ress.2015.10.015}}.
\newline\urlprefix\url{http://www.sciencedirect.com/science/article/pii/S0951832015002975}

\bibitem{Pronzato19}
L.~Pronzato, Sensitivity analysis via karhunen--lo{\`e}ve expansion of a random
  field model: Estimation of sobol’indices and experimental design,
  Reliability Engineering \& System Safety 187 (2019) 93--109.

\bibitem{HsingEubank15}
T.~Hsing, R.~Eubank, Theoretical foundations of functional data analysis, with
  an introduction to linear operators, John Wiley \& Sons, 2015.

\bibitem{Adler10}
R.~J. Adler, The geometry of random fields, SIAM, 2010.

\bibitem{DaPratoZabczyk14}
G.~Da~Prato, J.~Zabczyk, Stochastic equations in infinite dimensions,
  {Cambridge university press}, 2014.

\bibitem{big_Rudin}
W.~Rudin, Real and complex analysis, Tata McGraw-Hill Education, 1987.

\bibitem{Mercer1909}
J.~Mercer, Functions of positive and negative type, and their connection with
  the theory of integral equations, Philosophical Transactions of the Royal
  Society of London. Series A, Containing Papers of a Mathematical or Physical
  Character (1909) 415--446.

\bibitem{Lax02}
P.~D. Lax, Functional Analysis, John Wiley \& Sons, New-York, Chicester,
  Brisbane, Toronto, 2002.

\bibitem{hg}
J.~Hart, P.~Gremaud, An approximation theoretic perspective of {Sobol'} indices
  with dependent variables, International Journal for Uncertainty
  Quantification 8 (2018) 483--493.

\bibitem{SobolTarantolaGatelliEtAl06}
I.~Sobol', S.~Tarantola, D.~Gatelli, S.~Kucherenko, W.~Mauntz, Estimating the
  approximation error when fixing unessential factors in global sensitivity
  analysis, Reliability Engineering \& System Safety 92~(7) (2007) 957--960.

\bibitem{Ghanem:1991a}
R.~Ghanem, P.~Spanos, Stochastic Finite Elements: A Spectral Approach, Dover,
  2002, 2nd edition.

\bibitem{AlexanderianLeMaitrNajmEtAl12}
A.~Alexanderian, O.~L. Ma\^{i}tre, H.~Najm, M.~Iskandarani, O.~Knio, Multiscale
  stochastic preconditioners in non-intrusive spectral projection, Journal of
  Scientific Computing 50 (2012) 306--340.

\bibitem{ConradMarzouk13}
P.~R. Conrad, Y.~M. Marzouk, Adaptive {S}molyak pseudospectral approximations,
  SIAM Journal on Scientific Computing 35~(6) (2013) A2643--A2670.

\bibitem{WinokurConradSrajEtAl13}
J.~Winokur, P.~Conrad, I.~Sraj, O.~Knio, A.~Srinivasan, W.~C. Thacker,
  Y.~Marzouk, M.~Iskandarani, A priori testing of sparse adaptive polynomial
  chaos expansions using an ocean general circulation model database,
  Computational Geosciences 17~(6) (2013) 899--911.

\bibitem{WinokurKimBisettiEtAl16}
J.~Winokur, D.~Kim, F.~Bisetti, O.~P. Le~Ma{\^\i}tre, O.~M. Knio, Sparse pseudo
  spectral projection methods with directional adaptation for uncertainty
  quantification, Journal of Scientific Computing 68~(2) (2016) 596--623.

\bibitem{YanGuoXiu12}
L.~Yan, L.~Guo, D.~Xiu, Stochastic collocation algorithms using
  $\ell_1$-minimization, International Journal for Uncertainty Quantification
  2~(3) (2012).

\bibitem{HamptonDoostan16}
J.~Hampton, A.~Doostan, Compressive sampling methods for sparse polynomial
  chaos expansions, Handbook of Uncertainty Quantification (2017) 827--855.

\bibitem{FajraouiMarelliSudret17}
N.~Fajraoui, S.~Marelli, B.~Sudret, Sequential design of experiment for sparse
  polynomial chaos expansions, SIAM/ASA Journal on Uncertainty Quantification
  5~(1) (2017) 1061--1085.

\bibitem{DiazDoostanHampton18}
P.~Diaz, A.~Doostan, J.~Hampton, Sparse polynomial chaos expansions via
  compressed sensing and d-optimal design, Computer Methods in Applied
  Mechanics and Engineering 336 (2018) 640 -- 666.
\newblock \href {https://doi.org/https://doi.org/10.1016/j.cma.2018.03.020}
  {\path{doi:https://doi.org/10.1016/j.cma.2018.03.020}}.

\bibitem{spgl1:2007}
E.~van~den Berg, M.~P. Friedlander, "spgl1": A solver for large-scale sparse
  reconstruction, \url{http://www.cs.ubc.ca/labs/scl/spgl1}.

\bibitem{Smolyak63}
S.~Smolyak, Quadrature and interpolation formulas for tensor products of
  certain classes of functions, Dokl. Akad. Nauk SSSR 4 (1963) 240--243.

\bibitem{HeissWinschel08}
F.~Heiss, V.~Winschel, Likelihood approximation by numerical integration on
  sparse grids, Journal of Econometrics 144~(1) (2008) 62--80.

\bibitem{DexterTranWebster17}
N.~Dexter, H.~Tran, C.~Webster, \href{https://arxiv.org/abs/1711.02591}{On the
  strong convergence of forward-backward splitting in reconstructing jointly
  sparse signals}, arXiv preprint arXiv:1711.02591 (2017).
\newline\urlprefix\url{https://arxiv.org/abs/1711.02591}

\bibitem{GerritsmaVanderSteenVosEtAl10}
M.~Gerritsma, J.-B. Van~der Steen, P.~Vos, G.~Karniadakis, Time-dependent
  generalized polynomial chaos, Journal of Computational Physics 229~(22)
  (2010) 8333--8363.

\bibitem{PoetteLucor12}
G.~Po{\"e}tte, D.~Lucor, Non intrusive iterative stochastic spectral
  representation with application to compressible gas dynamics, Journal of
  Computational Physics 231~(9) (2012) 3587--3609.

\bibitem{OzenBal17}
H.~C. Ozen, G.~Bal, A dynamical polynomial chaos approach for long-time
  evolution of {SPDEs}, Journal of Computational Physics 343 (2017) 300--323.

\bibitem{ChuSudret17}
C.~V. Mai, B.~Sudret, Surrogate models for oscillatory systems using sparse
  polynomial chaos expansions and stochastic time warping, SIAM/ASA Journal on
  Uncertainty Quantification 5~(1) (2017) 540--571.

\bibitem{kunisch}
K.~{Kunisch}, S.~{Volkwein}, Galerkin proper orthogonal decomposition methods
  for a general equation in fluid dynamics, SIAM Journal on Numerical Analysis
  40 (2002) 492--515.

\bibitem{Constantine15}
P.~G. Constantine, Active subspaces: Emerging ideas for dimension reduction in
  parameter studies, SIAM, 2015.

\bibitem{TrefethenBau97}
L.~N. Trefethen, D.~Bau~III, Numerical linear algebra, Vol.~50, SIAM, 1997.

\bibitem{HalkoMartinssonTropp11}
N.~Halko, P.-G. Martinsson, J.~A. Tropp, Finding structure with randomness:
  Probabilistic algorithms for constructing approximate matrix decompositions,
  SIAM review 53~(2) (2011) 217--288.

\bibitem{petras:2003}
K.~Petras, Smolyak cubature of given polynomial degree with few nodes for
  increasing dimension, Numerische Mathematik 93~(4) (2003) 729--753.

\bibitem{LiIskandaraniLeHenaffEtAl16}
G.~Li, M.~Iskandarani, M.~Le~H{\'e}naff, J.~Winokur, O.~P. Le~Ma{\^\i}tre,
  O.~M. Knio, Quantifying initial and wind forcing uncertainties in the gulf of
  mexico, Computational Geosciences 20~(5) (2016) 1133--1153.

\bibitem{Hartley}
D.~M. Hartley, J.~G. Morris~Jr, D.~L. Smith, Hyperinfectivity: a critical
  element in the ability of v. cholerae to cause epidemics?, PLoS Med 3~(1)
  (2005) e7.

\end{thebibliography}

%
%
\appendix
\section{Definition of the index set $\mathcal{K}_i$}\label{sec:Ki}
Here we briefly explain the definition of the index set $\mathcal{K}_i$
in~\eqref{equ:Si_PCE}.
We use a multivariate PC basis that is 
obtained through partial tensorization of appropriate
univariate bases.  More precisely, if we denote by $\{\psi_j(\xi_i)\}_{j =
1}^\infty$ the univariate orthogonal polynomial basis corresponding to
$\xi_i$, we
form the multivariate PC basis $\{\Psi_k\}_{k = 0}^\infty$ as follows: 
\begin{equation}\label{equ:multivariatePC}
	\Psi_k(\xx) = \prod_{i=1}^{\Np} \psi_{\alpha_i^k}(\xi_i), \qquad \xx \in \Omega, \,
                      k = 0, 1, 2, \ldots, 
\end{equation}
where $\vec\alpha^k = \left( \alpha_1^k, \alpha_2^k, \cdots, \alpha_\Np^k\right)$, 
with $\alpha_k^i \geq 0$ for $i = 1, \ldots, \Np$, is
a multi-index, and $\alpha_i^k$ indicates the order of the 
univariate polynomials in
$\xi_i$.  In the present context, where 
$\xi_i$ is uniformly distributed,  $\psi_{\alpha_i^k}$
is the Legendre polynomial of order $\alpha_i^k$.
An indexing scheme for the 
multi-indices $\{\vec\alpha^k\}_{k = 0}^{\Np}$, which is convenient for 
computer implementations, can be found for instance in~\citep[Appendix C.]{LeMaitreKnio10}.
The multivariate orthogonal polynomial basis functions defined
in~\eqref{equ:multivariatePC} are used in the PC expansion of 
$f(t, \xx)$ in~\eqref{equ:PCE}.

The index set $\mathcal{K}_i$ in~\eqref{equ:Si_PCE}, which
picks all the terms in the PC
expansion that include $\xi_i$, is given by
$\mathcal{K}_i = \big\{ k \in  \{1, \ldots, \Npc\} : \alpha_k^i > 0\big\}$. 

\section{Proof of Theorem~\ref{thm:main}}\label{sec:proof}
First we establish some technical lemmas.
\begin{lem}\label{lem:var}
Let $f$ be a centered process  satisfying the assumptions of Section~\ref{sec:var} 
and let $\C$ be its covariance
operator with eigenpairs $\{(\lambda_i, e_i)\}_{i=1}^\infty$. The following hold:
\begin{enumerate}
\item  $\E{\fn} = 0$,
\item  $\var{\fn(t, \xx)} = \sum_{i = 1}^n \lambda_i e_i(t)^2$,
\item  We have
\[
   \lim_{n \to \infty} \int_0^T \var{ \fn(t, \xx)} \, dt = \int_0^T \var{f(t, \xx)} \, dt.
\]
\end{enumerate}
\end{lem}
\begin{proof}
The first statement is clear. The second statement is seen as follows:
\begin{multline}
\var{\fn(t, \xx)} = \E{ \fn(t, \xx)^2}
= \E{\left(\sum_{i = 1}^n f_i(\xx) e_i(t)\right)
     \left( \sum_{j = 1}^n f_j(\xx) e_j(t)\right)}
\\
= \sum_{i,j} \E{ f_i f_j} e_i(t) e_j(t) = \sum_{i = 1}^n \lambda_i e_i(t)^2.
\end{multline}
Finally, 
the third statement is derived as follows:
\begin{multline*}
\lim_{n \to \infty} \int_0^T \var{ \fn(t, \xx)} 
=\lim_{n \to \infty} \int_0^T \sum_{i = 1}^n \lambda_i e_i(t)^2 \, dt 
=\lim_{n \to \infty} \sum_{i = 1}^n \lambda_i \, dt\\ 
= \trace(\C) = 
\int_0^T c(t, t) \, dt = \int_0^T \var{f(t, \xx)}\, dt,
\end{multline*}
where the penultimate equality follows from Mercer's Theorem.
\end{proof}
\begin{lem}\label{lem:varj}
Let $f$ be a centered process  satisfying the assumptions of Section~\ref{sec:var}.
Then, for $U \subset \{1, \ldots, \Np\}$,
\[
\lim_{n \to \infty} \var{ \E{ \fn(t, \xx) | \xx_U }} = \var{ \E{ f(t, \xx) | \xx_U}}.
\]
\end{lem}
\begin{proof}
Let $t \in [0, T]$ be fixed but arbitrary. 
Since $\fn(t, \xx) \to f(t, \xx)$ in $L^2(\Omega)$, by properties of conditional expectation
$\E{ \fn(t, \xx) | \xx_U } \to \E{ f(t, \xx) | \xx_U }$ in $L^2(\Omega)$.
Next, noting that, $\E{\E{ \fn(t, \xx) | \xx_U}} = 0$, we
obtain
\[
\left| \var{  \E{ \fn(t, \xx) | \xx_U} } - \var{ \E{ f(t, \xx) | \xx_U} }\right|
= 
\left| \norm{  \E{ \fn(t, \xx) | \xx_U} }_{L^2(\Omega)}^2 - \norm{ \E{ f(t, \xx) | \xx_U} }_{L^2(\Omega)}^2\right| \to 0, 
\]
as $n \to \infty$.
\end{proof}

\textit{Proof of Theorem~\ref{thm:main}}.
We begin by considering $\lim_{n \to \infty} \int_0^T \var{ \E{\fn(t, \xx) | \xx_U}} \, dt$.
We have $\var{ \E{\fn(t, \xx) | \xx_U}} \leq \var{\fn(t, \xx)} 
\leq \var{f(t, \xx)}$, and $\var{f(t, \xx)} \in L^2(0, T)$. 
Therefore, by Lemma~\ref{lem:varj} and 
the Lebesgue Dominated 
Convergence Theorem,  
\[
   \lim_{n \to \infty} \int_0^T \var{ \E{\fn(t, \xx) | \xx_U}} \, dt
   =
   \int_0^T \var{ \E{ f(t, \xx) | \xx_U}} \, dt. 
\]
This, along with Lemma~\ref{lem:var}(3), yields
\begin{equation}\label{equ:lim_main}
\lim_{n \to \infty}
\frac{ \int_0^T \var{ \E{\fn(t, \xx) | \xx_U}} \, dt}
     { \int_0^T \var{ \fn(t, \xx) }\, dt}
= 
\frac{ \int_0^T \var{ \E{f(t, \xx) | \xx_U}} \, dt}
     { \int_0^T \var{ f(t, \xx) }\, dt}
=
\gen{S}^U(f; T).
\end{equation}

Next, note that,
$\E{ \fn(t, \xx) | \xx_U} 
  = \E{ \sum_{i=1}^n f_i(\xx) e_i(t) | \xx_U}
  = \sum_{i=1}^n \E{f_i(\xx) | \xx_U} e_i(t)$.
Then, we proceed as follows: 
\begin{align*}
\int_0^T \var{ \E{ \fn(t, \xx) | \xx_U}} \, dt 
&= \int_0^T \var{\sum_{i=1}^n \E{f_i(\xx) | \xx_U} e_i(t)} \, dt
\\
&= \int_0^T \E{\Big(\sum_{i=1}^n \E{f_i(\xx) | \xx_U} e_i(t)\Big)^2} \, dt
\\
&= \E{\int_0^T \Big(\sum_{i=1}^n \E{f_i(\xx) | \xx_U} e_i(t)\Big)^2 \, dt} 
\\
&= \E{\sum_{i,k=1}^n \E{f_i(\xx) | \xx_U} \E{f_k(\xx) | \xx_U} \int_0^T e_i(t)e_k(t)\, dt} 
\\
&= \E{\sum_{i=1}^n \E{f_i(\xx) | \xx_U}^2 } 
= \sum_{i = 1}^n \var{ \E{f_i(\xx) | \xx_U} }.
\end{align*}
Hence, using $\int_0^T \var{ \fn(t, \xx) }\, dt = \sum_{i=1}^n \lambda_i$, and~\eqref{equ:lim_main} yields
\[
\gen{S}^U(f; T) = 
\lim_{n\to\infty} 
\frac{ \int_0^T \var{ \E{\fn(t, \xx) | \xx_U}} \, dt}
     { \int_0^T \var{ \fn(t, \xx) }\, dt}
= 
\lim_{n\to\infty}    
\frac{ \sum_{i = 1}^n \var{ \E{f_i(\xx) | \xx_U} } } 
     { \sum_{i=1}^n \lambda_i}.~\square 
\]

\end{document}